\documentclass[conference]{IEEEtran}
\usepackage{eso-pic} 
\usepackage{algorithm}
\usepackage{algpseudocode}
\usepackage{amsmath}
\usepackage{amssymb}
\usepackage{url}
\usepackage{epsfig}
\usepackage{graphicx}
\usepackage{caption}
\usepackage{subcaption}

\makeatletter
\renewcommand{\ALG@beginalgorithmic}{\small}
\makeatother

\algnotext{EndIf}
\algnotext{EndFor}
\algnotext{EndWhile}

\newtheorem{theorem}{Theorem}[section]
\newtheorem{definition}{Definition}[section]

\ifCLASSINFOpdf
\else
\fi
\begin{document}
\sloppy
%
\title{Detecting Communities under Differential Privacy}

\author{\IEEEauthorblockN{Hiep H. Nguyen, Abdessamad Imine, and Micha\"{e}l Rusinowitch}
\IEEEauthorblockA{LORIA/INRIA Nancy-Grand Est, France}
Email: \{huu-hiep.nguyen,michael.rusinowitch\}@inria.fr, abdessamad.imine@loria.fr
}


%


\maketitle

\begin{abstract}
Complex networks usually expose community structure with groups of nodes sharing many links with the other nodes in the same group and relatively few with the nodes of the rest. This feature captures valuable information about the organization and even the evolution of the network. Over the last decade, a great number of algorithms for community detection have been proposed to deal with the increasingly complex networks. However, the problem of doing this in a private manner is rarely considered. 

In this paper, we solve this problem under differential privacy, a prominent privacy concept for releasing private data. We analyze the major challenges behind the problem and propose several schemes to tackle them from two perspectives: input perturbation and algorithm perturbation. We choose Louvain method as the back-end community detection for input perturbation schemes and propose the method LouvainDP which runs Louvain algorithm on a noisy super-graph. For algorithm perturbation, we design ModDivisive using exponential mechanism with the modularity as the score. We have thoroughly evaluated our techniques on real graphs of different sizes and verified their outperformance over the state-of-the-art.
\end{abstract}



%
\IEEEpeerreviewmaketitle

\section{Introduction}

\label{sec:introduction}

Graphs represent a rich class of data observed in daily life where entities are described by nodes and their connections are characterized by edges. Apart from microscopic (node level) and macroscopic (graph level) configurations, many complex networks 
display a \textit{mesoscopic} structure, i.e. they appear as a combination of components fairly independent of each other. These components are called communities, modules or clusters and the problem of how to reveal them plays a significant role in understanding the organization and function of complex networks. Over the last decade, a great number of algorithms for community detection (CD) have been proposed to address the problem in a variety of settings, such as undirected/directed, unweighted/weighted networks and non-overlapping/overlapping communities (for a comprehensive survey, see \cite{fortunato2010community}). 

These approaches, however, are adopted in a non-private manner, i.e. a data collector (such as Facebook) knows all the contributing users and their relationships before running CD algorithms. The output of such a CD, in the simplest form, is a clustering of nodes. Even in this case, i.e. only a node clustering (not the whole graph) is revealed, contributing users\ privacy may still be put at risk. 

In this paper, we address the problem of CD from the perspective of differential privacy \cite{dwork2006calibrating}. This privacy model offers a formal definition of privacy with a lot of interesting properties: no computational/informational assumptions about attackers, data type-agnosticity, composability and so on \cite{mcsherry2009privacy}. By differential privacy, we want to ensure the existence of connections between users to be hidden in the output clustering while keeping the low distortion of clusters compared to the ones generated by the corresponding non-private algorithms. 

As far as we know, the problem is quite new and only mentioned in the recent work \cite{mulle2015privacy} where M\"{u}lle et al. use a sampling technique to perturb the input graph so that it satisfies differential privacy before running conventional CD algorithms. This technique (we call it \textit{EdgeFlip} afterwards) is classified as \textit{input perturbation} in differential privacy literature (the other two categories are \textit{algorithm perturbation} and \textit{output perturbation}). Similarly, \textit{TmF} approach \cite{nguyen2015differentially} can apply to the true graph to get noisy output graphs as in the work of M\"{u}lle et al. Earlier, 1k-Series \cite{wang2013preserving}, \textit{Density Explore Reconstruct} (DER) \cite{chen2014correlated} and HRG-MCMC \cite{xiao2014differentially} are the best known methods for graph structure release under differential privacy. These methods can be followed by any exact CD algorithm to get a noisy clustering satisfying differential privacy. We choose Louvain method \cite{blondel2008fast} as such a CD algorithm. However, as we will see in the experiments, the output clusterings by the aforementioned methods have very low modularity. This fact necessitates new methods for CD problem under differential privacy.

Our main contributions are the new schemes \textit{LouvainDP} (input perturbation) and \textit{ModDivisive} (algorithm perturbation) which perform much better than the state-of-the-art. LouvainDP is a high-pass filtering method that randomly groups nodes into supernodes of equal size to build a weighted supergraph. LouvainDP is guaranteed to run in linear time. ModDivisive is a top-down approach which privately divides the node set into the k-ary tree guided by the modularity score at each level. The main technique used in ModDivisive is the Markov Chain Monte Carlo (MCMC) to realize the exponential mechanism \cite{mcsherry2007mechanism}. We show that ModDivisive's runtime is linear in the number of nodes, the height of the binary tree and the burn-in factor of MCMC. The linear complexity enables us to examine million-scale graphs in a few minutes. The experiments show the high modularity and low distortion of the output clusters by LouvainDP and ModDivisive.

Our contributions are summarized as follows:
\begin{itemize}
\item We analyze the major challenges of community detection under differential privacy. We explain why techniques borrowed from k-Means fail and how the difficulty of $\epsilon$-DP recommender systems justifies a relaxation of $\epsilon$.
\item We design an input perturbation scheme LouvainDP that runs in linear time using the high-pass filtering technique from \cite{cormode2012differentially} and Louvain method \cite{blondel2008fast}.
\item We propose an algorithm perturbation scheme ModDivisive as a divisive approach by using the modularity-based score function in the exponential mechanism. We prove that modularity has small global sensitivity and ModDivisive also runs in linear time.
\item We conduct a thorough evaluation on real graphs of different sizes and show the outperformance of LouvainDP and ModDivisive over the state-of-the-art.
\end{itemize}

The paper is organized as follows. We review the related work for community detection algorithms and graph release via differential privacy in Section \ref{sec:related}. Section \ref{sec:pre} briefly introduces the key concepts of differential privacy, the popular Louvain method and the major challenges of $\epsilon$-DP community detection. Section \ref{sec:input} focuses on the category of input perturbation in which we propose LouvainDP and review several recent input perturbation schemes. We describe ModDivisive in Section \ref{sec:algo}. We compare all the presented schemes on real graphs in Section \ref{sec:eval}. Finally, we present our conclusions and suggest future work in Section \ref{sec:conclusion}.
 
 Table \ref{tab:notation} summarizes the key notations used in this paper.
 
 \begin{table}[htb]
 \small
 \centering
 \caption{List of notations} \label{tab:notation}
 \begin{tabular}{c|l}
 \textbf{Symbol} &\textbf{Definition} \\
 \hline
 $G=(V,E_G)$ & true graph with $n=|V|$ and $m=|E_G|$\\
 \hline
 $G'=(V,E_{G'})$ & neighboring graph of $G$\\
 \hline
 $\tilde{G}=(V,E_{\tilde{G}})$ & sample noisy output graph\\
 \hline
 $G_1=(V_1,E_1)$ & supergraph generated by LouvainDP\\ 
 \hline
 $k$ & fan-out of the tree in ModDivisive\\
 \hline
 $K$ & burn-in factor in MCMC-based algorithms \\
 \hline
 $\lambda$ & common ratio to distribute the privacy budget\\
 \hline
  $C$ & a clustering of nodes in $G$  \\
 \hline
  $Q(G, C)$ & modularity of the clustering $C$ on graph $G$  \\
 \end{tabular}
 \end{table}
 
\section{Related Work}
\label{sec:related}

\subsection{Community Detection in Graphs}
There is a vast literature on community detection in graphs. For a recent comprehensive survey, we refer to \cite{fortunato2010community}. In this section, we discuss several classes of techniques.

Newman and Girvan \cite{newman2004finding} propose \textit{modularity} as a quality of network clustering. It is based on the idea that a random graph is not expected to have a modular structure, so the possible existence of clusters is revealed by the comparison between the actual density of edges in a subgraph and the density one would expect to have in the subgraph if the nodes of the graph were connected randomly (the null model). The modularity $Q$ is defined as

\begin{equation}
\label{eqn:modularity}
Q = \sum_{c=1}^{n_c}\left[\frac{l_c}{m} - \left(\frac{d_c}{2m} \right) ^2 \right]
\end{equation}
where $n_c$ is the number of clusters, $l_c$ is the total number of edges joining nodes in community $c$ and $d_c$ is the sum of the degrees of the nodes of $c$. 

Many methods for optimizing the modularity have been proposed over the last ten years, such as agglomerative greedy \cite{clauset2004finding},  simulated annealing \cite{medus2005detection}, random walks \cite{pons2005computing}, statistical mechanics \cite{reichardt2006statistical}, label propagation \cite{raghavan2007near} or InfoMap \cite{rosvall2008maps}, just to name a few. The recent multilevel approach, also called \textit{Louvain method}, by Blondel et al. \cite{blondel2008fast} is a top performance scheme. It scales very well to graphs with hundreds of millions of nodes. This is the chosen method for the input perturbation schemes considered in this paper (see Sections \ref{subsec:louvain} and \ref{sec:input} for more detail). By maximizing the modularity, Louvain method is based on \textit{edge counting} metrics, so it fits well with the concept of \textit{edge differential privacy} (Section \ref{sec:diffpriv}). One of the most recent methods \textit{SCD} \cite{prat2014high} is not chosen because it is about maximizing Weighted Community Clustering (WCC) instead of the modularity. WCC is based on \textit{triangle counting} which has high global sensitivity (up to $O(n)$) \cite{zhang2015private}. Moreover, SCD pre-processes the graph by removing all edges that do not close any triangle. This means all 1-degree nodes are excluded and form singleton clusters. The number of output clusters is empirically up to $O(n)$.

\subsection{Graph Release via Differential Privacy}
In principle, after releasing a graph satisfying $\epsilon$-DP, we can do any mining operations on it, including community detection. The research community, therefore, expresses a strong interest in the problem of graph release via differential privacy. Differentially private algorithms relate the amount of noise to the computation sensitivity. Lower sensitivity implies smaller added noise. Because the edges in simple undirected graphs are usually assumed to be independent, the standard Laplace mechanism \cite{dwork2006calibrating} is applicable (e.g. adding Laplace noise to each cell of the adjacency matrix). However, this approach severely deteriorates the graph structure. 

The state-of-the-art \cite{chen2014correlated,xiao2014differentially} try to reduce the sensitivity of the graph in different ways. \textit{Density Explore Reconstruct} (DER)\cite{chen2014correlated} employs a data-dependent quadtree to summarize the adjacency matrix into a counting tree and then reconstructs noisy sample graphs. DER is an instance of input perturbation. Xiao et al. \cite{xiao2014differentially} propose to use \textit{Hierarchical Random Graph} (HRG) \cite{clauset2008hierarchical} to encode graph structural information in terms of edge probabilities. Their scheme HRG-MCMC is argued to be able to sample good HRG models which reflect the community structure in the original graph. HRG-MCMC is classified as an algorithm perturbation scheme. A common disadvantage of the state-of-the-art DER and HRG-MCMC is the scalability issue. Both of them incur quadratic complexity $O(n^2)$, limiting themselves to medium-sized graphs.

Recently, Nguyen et al. \cite{nguyen2015differentially} proposed \textit{TmF} that utilizes a filtering technique to keep the runtime linear in the number of edges. TmF also proves the upper bound $O(\ln n)$ for the privacy budget $\epsilon$. At the same time, M\"{u}lle et al. \cite{mulle2015privacy} devised the scheme \textit{EdgeFlip} using an edge flipping technique. However, EdgeFlip costs $O(n^2)$ and is runnable only on graphs of tens of thousands of nodes.

A recent paper by Campan et al. \cite{campan2015preserving} studies whether anonymized social networks preserve existing communities from the original social networks. The considered anonymization methods are SaNGreeA \cite{campan2008clustering} and k-degree method \cite{liu2008towards}. Both of them use k-anonymity rather than differential privacy which is the focus of this paper.

\section{Preliminaries}
\label{sec:pre}
In this section, we review key concepts and mechanisms of differential privacy.
\subsection{Differential Privacy}
\label{sec:diffpriv}
Essentially, $\epsilon$-differential privacy ($\epsilon$-DP) \cite{dwork2006calibrating} is proposed to quantify the notion of \textit{indistinguishability} of neighboring databases. In the context of graph release, two graphs $G_1=(V_1,E_1)$ and $G_2=(V_2,E_2)$ are neighbors if $V_1=V_2$, $E_1 \subset E_2$ and $|E_2| = |E_1| + 1$. Note that this notion of neighborhood is called \textit{edge differential privacy} in contrast with the notion of \textit{node differential privacy} which allows the addition of one node and its adjacent edges. Node differential privacy has much higher sensitivity and is much harder to analyze \cite{kasiviswanathan2013analyzing,blocki2013differentially}. Our work follows the common use of edge differential privacy. The formal definition of $\epsilon$-DP for graph data is as follows.

\begin{definition}
\label{def:e-dp}
A mechanism $\mathcal{A}$ is $\epsilon$-differentially private if for any two neighboring graphs $G_1$ and $G_2$, and for any output $O \in Range(\mathcal{A})$,
\begin{equation}
Pr[\mathcal{A}(G_1) \in O] \leq e^{\epsilon} Pr[\mathcal{A}(G_2) \in O] \nonumber
\end{equation} 
\end{definition}

Laplace mechanism \cite{dwork2006calibrating} and Exponential mechanism \cite{mcsherry2009privacy} are two standard techniques in differential privacy. The latter is a generalization of the former. Laplace mechanism is based on the concept of \textit{global sensitivity} of a function $f$ which is defined as $\Delta f = \max_{G_1,G_2} ||f(G_1)-f(G_2)||_1$ where the maximum is taken over all pairs of neighboring $G_1,G_2$. Given a function $f$ and a privacy budget $\epsilon$, the noise is drawn from a Laplace distribution $Lap(\lambda): p(x|\lambda) = \frac{1}{2\lambda} e^{-|x|/\lambda}$ where $\lambda=\Delta f/\epsilon$.

\begin{theorem} (Laplace mechanism \cite{dwork2006calibrating})
\label{def:lap-mech}
For any function $f: G \rightarrow \mathbb{R}^d$, the mechanism $\mathcal{A}$
\begin{equation}
\mathcal{A}(G) = f(G) + \langle Lap_1(\frac{\Delta f}{\epsilon}),...,Lap_d(\frac{\Delta f}{\epsilon})\rangle
\end{equation}
satisfies $\epsilon$-differential privacy, where $Lap_i(\frac{\Delta f}{\epsilon})$ are i.i.d Laplace variables with scale parameter $\frac{\Delta f}{\epsilon}$.	\hfill $\Box$
\end{theorem}

\textit{Geometric mechanism} \cite{ghosh2012universally} is a discrete variant of Laplace mechanism with integral output range $\mathbb{Z}$ and random noise $\Delta$ generated from a  two-sided geometric distribution $Geom(\alpha): Pr[\Delta = \delta| \alpha] = \frac{1-\alpha}{1+\alpha} \alpha^{|\delta|}$. To satisfy $\epsilon$-DP, we set $\alpha=\exp(-\epsilon)$. We use geometric mechanism in our LouvainDP scheme.

For non-numeric data, the exponential mechanism is a better choice \cite{mcsherry2007mechanism}. Its main idea is based on sampling an output $O$ from the output space $\mathcal{O}$ using a score function $u$. This function assigns exponentially higher probabilities to outputs of higher scores. Let the global sensitivity of $u$ be $\Delta u = \max_{O,G_1,G_2}|u(G_1,O) - u(G_2,O)|$.

\begin{theorem} (Exponential mechanism \cite{mcsherry2007mechanism})
\label{def:exp-mech}
Given a score function $u:(G\times \mathcal{O}) \rightarrow \mathbb{R}$ for a graph $G$, the mechanism $\mathcal{A}$ that samples an output $O$ with probability proportional to $\exp(\frac{\epsilon .u(G,O)}{2\Delta u})$ satisfies $\epsilon$-differential privacy.		\hfill $\Box$
\end{theorem}

Composability is a nice property of differential privacy which is not satisfied by other privacy models such as  k-anonymity.

\begin{theorem} (Sequential and parallel compositions \cite{mcsherry2009privacy})
\label{theorem:composition}
Let each $A_i$  provide $\epsilon_i$-differential privacy. A sequence of $A_i(D)$ over the dataset D provides $\Sigma_{i=1}^n \epsilon_i$-differential privacy.

Let each  $A_i$  provide $\epsilon_i$-differential privacy. Let $D_i$ be arbitrary disjoint subsets of the dataset D. The sequence of $A_i(D_i)$ provides $\max_{i=1}^n\epsilon_i$-differential privacy. \hfill $\Box$
\end{theorem}

\subsection{Louvain Method}
\label{subsec:louvain}
Since its introduction in 2008, Louvain method \cite{blondel2008fast} becomes one of the most cited methods for the community detection task. It optimizes the modularity by a bottom-up folding process. The algorithm is divided in passes each of which is composed of two phases that are repeated iteratively. Initially, each node is assigned to a different community. So, there will be as many communities as there are nodes in the first phase. Then, for each node $i$, the method considers the gain of modularity if we move $i$ from its community to the community of a neighbor $j$ (a \textit{local change}). The node $i$ is then placed in the community for which this gain is maximum and positive (if any), otherwise it stays in its original community. This process is applied repeatedly and sequentially for all nodes until no further improvement can be achieved and the first pass is then complete.

\begin{figure}
\centering
\includegraphics[height=1.6in]{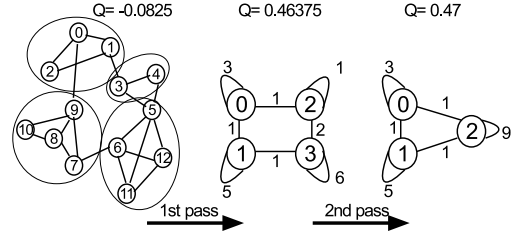}
\setlength{\abovecaptionskip}{-10pt}
\caption{Louvain method}
\vspace{-1.0em}
\label{fig:louvain}
\end{figure}

We demonstrate Louvain method in Fig.\ref{fig:louvain} by a graph of 13 nodes and 20 edges. If each node forms its own singleton community, the modularity Q will be -0.0825. In the first pass of Louvain method, each node moves to the best community selected from its neighbors' communities. We get the partition $[\{0,1,2\}, \{3,4\}, \{5,6,11,12\}, \{7,8,9,10\}]$ with modularity 0.46375. The second phase of first pass builds a weighted graph corresponding to the partition by aggregating communities. The second pass repeats the folding process on this weighted graphs to reach the final partition $[\{0,1,2\}, \{3,4,5,6,11,12\}, \{7,8,9,10\}]$ with modularity 0.47.

This simple agglomerative algorithm has several advantages as stated in \cite{blondel2008fast}. First, its steps are intuitive and easy to implement, and the outcome is unsupervised. Second, the algorithm is extremely fast, i.e. computer simulations on large modular networks suggest that its complexity is linear on typical and sparse data. This is due to the fact that the possible gains in modularity are easy to compute and the number of communities decreases drastically after just a few passes so that most of the running time is concentrated on the first iterations. Third, the multi-level nature of the method produces a hierarchical structure of communities which allows multi-resolution analysis, i.e the user can zoom in the graph to observe its structure with the desired resolution. In addition, Louvain method is runnable on weighted graphs. This fact supports naturally our scheme LouvainDP as described in Section \ref{subsec:louvaindp}.

\subsection{Challenges of Community Detection under Differential Privacy}
\label{subsec:challenge}

In this section, we explain why community detection under differential privacy is challenging. We show how techniques borrowed from related problems fail. We also advocate the choice of $\epsilon$ as a function of graph size $n$.

The problem of differentially private community detection is closely related to $\epsilon$-DP k-Means clustering and recommender systems. The $\epsilon$-DP k-Means is thoroughly discussed in \cite{su2015differentially}. However, techniques from $\epsilon$-DP k-Means are not suitable to $\epsilon$-DP community detection. First, items in k-Means are in low-dimensional spaces and the number of clusters $k$ is usually small. This  contrast to the case of community detection where nodes lie in a $n$-dimensional space and the number of communities varies from tens to tens of thousands, not to say the communities may overlap or be nested (multi-scale). Second, items in $\epsilon$-DP k-Means are normalized to $\left[-1,1\right] ^d$ while the same preprocessing seems invalid in $\epsilon$-DP community detection. Moreover, the output of k-Means usually consists of equal-sized balls while this is not true for communities in graphs. Considering the graph as a high-dimensional dataset, we tried the private projection technique in \cite{kenthapadi2012privacy} which is followed by spectral clustering, but the modularity scores of the output are not better than random clustering.

Recent papers on $\epsilon$-DP recommender systems \cite{guerraoui2015d,banerjee2015price} show that privately learning the clustering of items from user ratings is hard unless we relax the value of $\epsilon$ up to $\log n$. Banerjee et al. \cite{banerjee2015price} model differentially private mechanisms as noisy channels and bound the mutual information between the generative sources and the privatized sketches. They show that in the information-rich regime (each user rates $O(n)$ items), their \textit{Pairwise-Preference} succeeds if the number of users is $\Omega(n\log n /\epsilon)$. Compared to $\epsilon$-DP community detection where the number of users is $n$, we should have $\epsilon=\Omega(\log n)$. Similarly, in D2P scheme, Guerraoui et al. \cite{guerraoui2015d} draw a formula for $\epsilon$ as
\begin{equation}
\epsilon_{D2P}^{(p,0,\lambda)} = \ln(1+\frac{(1-p).\mathcal{N}_E}{p.|\mathcal{G}_\lambda|})
\end{equation}
where $\lambda$ is the distance used to conceal the user profiles ($\lambda=0$ reduces to the classic notion of differential privacy). $\mathcal{N}_E$ is the number of items which is exactly $n$ in community detection. $|\mathcal{G}_\lambda|$ is the minimum size of user profiles at distance $\lambda$ over all users ($|\mathcal{G}_\lambda|$ = $o(\mathcal{N}_E)$ except at unreasonably large $\lambda$). Clearly, at $p = 0.5$ (as used in \cite{guerraoui2015d}), we have $\epsilon_{D2P}^{(0.5,0,\lambda)} \approx \ln n$. The sampling technique in D2P is very similar to EdgeFlip \cite{mulle2015privacy} which is shown ineffective in $\epsilon$-DP community detection for $\epsilon \in  (0,0.5\ln n)$ (Section \ref{sec:eval}). Note that $\epsilon$-DP community detection is unique in the sense that the set of items and the set of users are the same. Graphs for community detection are more general than bipartite graphs in recommender systems. In addition, modularity $Q$ (c.f. Formula \ref{eqn:modularity}) is \textit{non-monotone}, i.e. for two disjoint sets of nodes $A$ and $B$, $Q(A \cup B)$ may be larger, smaller than or equal to $Q(A) + Q(B)$.

To further emphasize the difficulty of $\epsilon$-DP community detection, we found that IDC scheme \cite{gupta2012iterative} using \textit{Sparse Vector Technique} \cite[Section 3.6]{dwork2014algorithmic} is hardly feasible. As shown in Algorithm 1 of \cite{gupta2012iterative}, to publish a noisy graph that can approximately answer all cut queries with bounded error $m^{0.25}n/\epsilon^{0.5}$, IDC must have $B(\alpha)$ ``yes'' queries among all $k$ queries. $B(\alpha)$ may be as low as $\sqrt{m}$ but $k = 2^{2n}$. In the average case, IDC incurs exponential time to complete.

To conclude, $\epsilon$-DP community detection is challenging and requires new techniques. In this paper, we evaluate the schemes for $\epsilon$ up to $0.5\ln n$. At $\epsilon = 0.5\ln n$, the multiplicative ratio (c.f. Definition \ref{def:e-dp}) is $e^\epsilon = e^{0.5\ln n} = \sqrt{n}$, a reasonable threshold for privacy protection compared to $\epsilon=\ln n$ (i.e. $e^\epsilon = n$) in $\epsilon$-DP recommender systems discussed above. The typical epsilon in the literature is 1.0 or less. However, this value is only applicable to graph metrics of low sensitivity $O(1)$ such as the number of edges, the degree sequence. The global sensitivity of other metrics like the diameter, the number of triangles, 2K-series etc. is $O(n)$, calling for local sensitivity analysis (e.g. \cite{nissim2007smooth}). For counting queries, Laplace/Geometric mechanisms are straightforward on real/integral (\textit{metric}) spaces. However, \textit{direct} noise adding mechanisms on the space $\mathcal{P}$ of all ways to partition the nodeset $V$ are non-trivial because $|\mathcal{P}| \approx n^n$ and $\mathcal{P}$ is non-metric. That is another justification for the relaxation of $\epsilon$ up to $0.5\ln n$.


\section{Input Perturbation}
\label{sec:input}
In this section, we propose the linear scheme LouvainDP that uses a filtering technique to build a noisy weighted supergraph and calls the exact Louvain method subsequently. Then we discuss several recent $\epsilon$-DP schemes that can be classified as input perturbation. Fig.\ref{fig:input-perturb} sketches the basic steps of the input perturbation paradigm.

\begin{figure}
	\begin{center}
        \begin{subfigure}[b]{0.42\textwidth}
                \centering
                \includegraphics[height=0.9in]{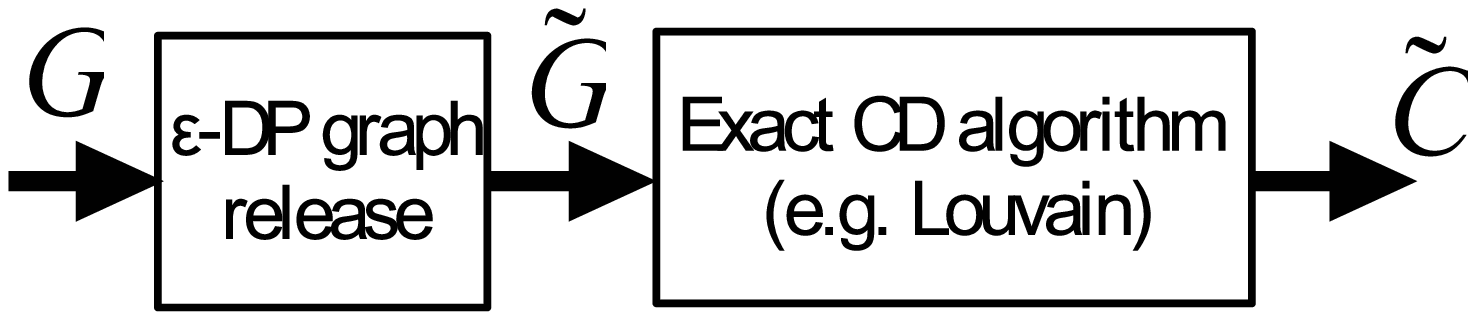}
                \setlength{\abovecaptionskip}{-20pt}
               	\captionof{figure}{Input perturbation}
                \label{fig:input-perturb}
        \end{subfigure}        
        \begin{subfigure}[b]{0.42\textwidth}
                \centering
                \includegraphics[height=0.9in]{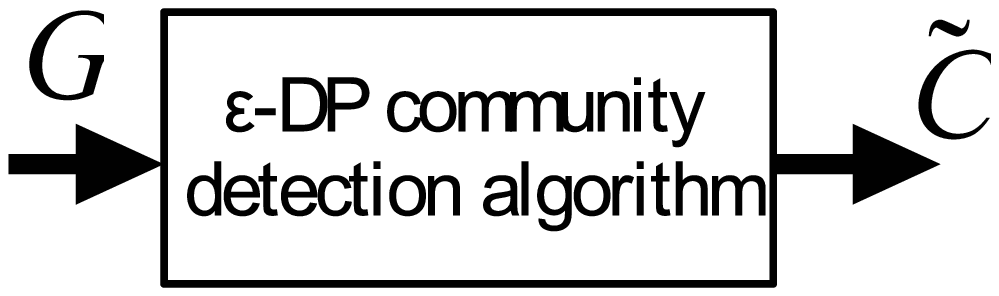}
                \setlength{\abovecaptionskip}{-20pt}
               	\captionof{figure}{Algorithm perturbation}
               	\label{fig:algo-perturb}
        \end{subfigure}
    \end{center}
    \setlength{\abovecaptionskip}{-5pt}
	\caption{Two categories of $\epsilon$-DP community detection}
    \vspace{-1.5em}
\end{figure}

\subsection{LouvainDP: Louvain Method on Noisy Supergraphs}
\label{subsec:louvaindp}
The basic idea of LouvainDP is to create a noisy weighted supergraph $G_1$ from $G$ by grouping nodes into supernodes of equal size $k$. We then apply the filtering technique of Cormode et al. \cite{cormode2012differentially} to ensure only $O(m)$ noisy weighted edges in $G_1$. Finally, we run the exact community detection on $G_1$.

In \cite{cormode2012differentially}, Cormode et al. propose several summarization techniques for sparse data under differential privacy. Let $M$ be a contingency table having the domain size $m_0$ and $m_1$ non-zero entries ($m_1 \ll m_0$ for sparse data), the conventional publication of a noisy table $M'$ from $M$ that satisfies $\epsilon$-DP requires the addition of Laplace/geometric noise to $m_0$ entries. The entries in $M'$ could be filtered (e.g. removing negative ones) and/or sampled to get a noisy summary $M''$. This direct approach would be infeasible for huge domain sizes $m_0$. Techniques in \cite{cormode2012differentially} avoid materializing the vast noisy data by computing the summary $M''$ directly from $M$ using filtering and sampling techniques.

In our LouvainDP, the supergraph $G_1$ is an instance of sparse data with the domain size $m_0 = \frac{n_1(n_1+1)}{2}$ where $n_1$ is the number of supernodes and $m_1$ non-zero entries corresponding to non-zero superedges. We use the one-sided filtering \cite{cormode2012differentially} to efficiently compute $G_1$ with $O(m)$ edges in linear time.

\subsubsection{Algorithm}
LouvainDP can run with either geometric or Laplace noise. We describe the version with geometric noise in Algorithm \ref{algo-louvaindp}.

\begin{algorithm}
\caption{LouvainDP($G,s$)}
\label{algo-louvaindp}
\begin{algorithmic}[1]
	\Require undirected graph $G$, group size $k$, privacy budget $\epsilon$
	\Ensure noisy partition $\tilde{C}$
	\State $G_1 \leftarrow \varnothing$, $n_1 = \lfloor \frac{|V|}{k}\rfloor -1$,  $V_1 \leftarrow \{0,1,..,n_1 \}$
	\State $\epsilon_2 = 0.1$, $\epsilon_1 = \epsilon - \epsilon_2$, $\alpha = \exp(-\epsilon_1)$
	\State get a random permutation $V_p$ of $V$
	\State compute the mapping $M: V_p \rightarrow V_1$
	\State compute superedges of $G_1$: $E_1 = \{e_1(i,j)\}$ where $i,j \in V_1$
	\State $m_1 = |E_1| + Lap(1/\epsilon_2)$, $m_0 = \frac{n_1(n_1+1)}{2}$
	\State $\theta = \lceil \log_{\alpha} \frac{(1+\alpha)m_1}{m_0-m_1} \rceil $	
	\State $s = (m_0 - m_1) \frac{\alpha^{\theta}}{1 + \alpha}$
	\For {$e_1(i,j)$ in $E_1$ }
		\State $e_1(i,j) = e_1(i,j) + Geom(\alpha)$
		\If {$e_1(i,j) \geq \theta$}
			\State add $e_1(i,j)$ to $G_1$
		\EndIf
	\EndFor
	\For {$s$ edges sampled uniformly at random $e_1(i,j) \notin E_1$}
		\State draw $w$ from the distribution $Pr[X \leq x] = 1 - \alpha^{x - \theta + 1}$
		\If {$w > 0$}
			\State add edge $e_1(i,j)$ with weight $w$ to $G_1$
		\EndIf
	\EndFor
	\State run Louvain method on $G_1$ to get $\tilde{C}_1$
	\State compute $\tilde{C}$ from $\tilde{C}_1$ using the mapping $M$ \\
	\Return $\tilde{C}$
\end{algorithmic}
\end{algorithm}

Given the group size $k$, LouvainDP starts with a supergraph $G_1$ having $\lfloor \frac{|V|}{k}\rfloor$ nodes by randomly permuting the nodeset $V$ and grouping every $k$ consecutive nodes into a supernode (lines 1-4). The permutation prevents the possible bias of node ordering in $G$. The set of superedges $E_1$ is easily computed from G. Note that $m_1 = |E_1| \leq m$ due to the fact that each edge of $G$ appears in one and only one superedge. The domain size is $m_0 = \frac{n_1(n_1+1)}{2}$ (i.e. we consider all selfloops in $G_1$). The noisy number of non-zero superedges is $m_1 = |E_1| + Lap(1/\epsilon_2)$. Then by \textit{one-sided} filtering \cite{cormode2012differentially}, we estimate the threshold $\theta$ (line 7) and the number of passing zero superedges $s$ (line 8). For each non-zero superedge, we add a geometric noise and add the superedge to $G_1$ if the noisy value is not smaller than $\theta$. For $s$ zero superedges $e_1(i,j) \notin E_1$, we draw an integral weight $w$ from the distribution $Pr[X \leq x] = 1 - \alpha^{x - \theta + 1}$ and add $e_1(i,j)$ with weight $w$ to $G_1$ if $w > 0$.

\subsubsection{Complexity}
LouvainDP runs in $O(m)$ because the loops to compute superedges (Line 5) and to add geometric noises (lines 9-12) cost $O(m)$. We have $s = (m_0 - m_1) \frac{\alpha^{\theta}}{1 + \alpha} \leq \frac{m_0 - m_1}{1+\alpha} \frac{(1+\alpha)m_1}{m_0-m_1} = m_1$ (see Line 7). So the processing of $s$ zero-superedges costs $O(m)$. Moreover, Louvain method (line 17) is empirically linear in $m_1$ \cite{blondel2008fast}. We come up with the following theorem.

\begin{theorem}
\label{lem:louvaindp-edges}
The number of edges in $G_1$ is not larger than $2m$. LouvainDP's runtime is $O(m)$
\end{theorem}

\subsubsection{Privacy Analysis}
In LouvainDP, we use a small privacy budget $\epsilon_2 = 0.01$ to compute the noisy number of non-zero superedges $m_1$. The remaining privacy budget $\epsilon_1$ is used for the geometric mechanism $Geom(\alpha)$. Note that getting a random permutation $V_p$ (line 3) costs no privacy budget.
The number of nodes $n$ is public and given the group size $k$, the number of supernodes $n_1$ is also public. The high-pass filtering technique (Lines 6-16) inherits the privacy guarantee by \cite{cormode2012differentially}. By setting $\epsilon_1 = \epsilon - \epsilon_2$, LouvainDP satisfies $\epsilon$-differential privacy (see the sequential composition (Theorem \ref{theorem:composition})).

\subsection{Alternative Input Perturbation Schemes}
1K-series \cite{wang2013preserving}, DER \cite{chen2014correlated}, TmF \cite{nguyen2015differentially} and EdgeFlip \cite{mulle2015privacy} are the most recent differentially private schemes for graph release that can be classified as input perturbation. While 1K-series and TmF run in linear time, DER and EdgeFlip incur a quadratic complexity. DER and EdgeFlip are therefore tested only on two medium-sized graphs in Section \ref{sec:eval}. 

The expected number of edges by EdgeFlip is $|E_{\tilde{G}}| = (1-s)m + \frac{n(n-1)}{4} s$ (see \cite{mulle2015privacy}) where $s = \frac{2}{e^{\epsilon} + 1}$ is the flipping probability. Substitute $s$ into $|E_{\tilde{G}}|$, we get $ |E_{\tilde{G}}| = m + (\frac{n(n-1)}{4} - m)\frac{2}{e^{\epsilon} + 1}$. The number of edges in the noisy graph $\tilde{G}$ generated by EdgeFlip increases exponentially as $\epsilon$ decreases. To ensure the linear complexity for million-scale graphs, we propose a simple extension of EdgeFlip, called \textit{EdgeFlipShrink} (Algorithm \ref{algo-edgeflipshrink}). 

Instead of outputting $\tilde{G}$, EdgeFlipShrink computes $\hat{G}$ that has the expected number of edges $m$ by shrinking $E_{\tilde{G}}$. First, the algorithm compute the private number of edges $\tilde{m}$ using a small budget $\epsilon_2$ (Lines 2-3 ). The new flipping probability $\tilde{s}$ is updated (Line 5). The noisy expected number of edges in the original EdgeFlip is shown in Line 6. We obtain the shrinking factor $p = \frac{\tilde{m}}{m_0}$ (Line 7). Using $p$, every 1-edge is sampled with probability $\frac{1-\tilde{s}}{2} p$ instead of $\frac{1-s}{2}$ as in \cite{mulle2015privacy}. The remaining 0-edges are randomly picked from $E_G$ as long as they do not exist in $\hat{G}$ (Lines 14-19).

The expected edges of $\hat{G}$ is $E[\hat{G}] = E[\tilde{m}] = m$ and the running time of EdgeFlipShrink is $O(m)$.

\begin{algorithm}
\caption{EdgeFlipShrink($G,s$)}
\label{algo-edgeflipshrink}
\begin{algorithmic}[1]
	\Require undirected graph $G$, flipping probability $s$
	\Ensure anonymized graph $\hat{G}$
	\State $\hat{G} \leftarrow \varnothing$	
	\State $\epsilon_2 = 0.1$
	\State $\tilde{m} = m + Lap(1/\epsilon_2)$
	\State $\epsilon = \ln(\frac{2}{s}-1) - \epsilon_2$	
	\State $\tilde{s} = \frac{2}{e^\epsilon + 1}$
	\State $m_0 = (1-\tilde{s})\tilde{m} + \frac{n(n-1)}{4}\tilde{s}$
	\State $p = \frac{\tilde{m}}{m_0}$	
	\State // process 1-edges
	\State $n_1 = 0$
	\For {edge $(i,j) \in E_G$}
		\State add edge $(i,j)$ to $\hat{G}$ with prob. $\frac{1-\tilde{s}}{2} p$	
		\State $n_1 ++$
	\EndFor	
	\State // process 0-edges
	\State $n_0 = \tilde{m} - n_1$
	\While {$n_0 > 0$}
		\State random pick an edge $(i,j) \notin E_G$
		\If {$\hat{G}$ does not contain $(i,j)$}
			\State add edge $(i,j)$ to $\hat{G}$
			\State $n_0$- -
		\EndIf
	\EndWhile \\
	\Return $\hat{G}'$
\end{algorithmic}
\end{algorithm}

\section{Algorithm Perturbation}
\label{sec:algo}
The schemes in the algorithm perturbation category privately sample a node clustering from the input graph without generating noisy sample graphs as in the input perturbation. This can be done via the exponential mechanism. We introduce our main scheme \textit{ModDivisive} in Section \ref{subsec:moddivisive} followed by a variant of HRG-MCMC runnable on large graphs (Section \ref{subsec:hrg-fixed}). Fig.\ref{fig:algo-perturb} sketches the basic steps of the algorithm perturbation paradigm.

\subsection{ModDivisive: Top-down Exploration of Cohesive Groups}
\label{subsec:moddivisive}
\subsubsection{Overview}
In contrast with the agglomerative approaches (e.g. Louvain method) in which small communities are iteratively merged if doing so increases the modularity, our ModDivisive is a divisive algorithm in which communities at each level are iteratively split into smaller ones. Our goal is to heuristically detect cohesive groups of nodes in a private manner. There are several technical challenges in this process. The first one is how to efficiently find a good split of nodes that induces a high modularity and satisfies $\epsilon$-DP at the same time. The second one is how to merge the small groups to  larger ones. 
We cope with the first challenge by realizing an exponential mechanism via MCMC (Markov Chain Monte-Carlo) sampling with the modularity as the score function (see Theorem \ref{def:exp-mech}). The second challenge is solved by dynamic programming.
We design ModDivisive as a \textit{k-ary} tree (Fig.\ref{fig:mod-divisive}), i.e. each internal node has no more than $k$ child nodes. The root node (level 0) contains all nodes in $V$ and assigns arbitrarily each node into one of the $k$ groups. Then we run the MCMC over the space of all partitions of $V$ into no more than $k$ nonempty subsets. The resultant subsets are initialized as the child nodes (level 1) of the root. The process is repeated iteratively for each child node at level 1 and stops at level $maxL$. Fig.\ref{fig:mod-divisive} illustrates the idea with $k=3$ for the graph in Fig.\ref{fig:louvain}.
\begin{figure}
\centering
\includegraphics[height=1.4in]{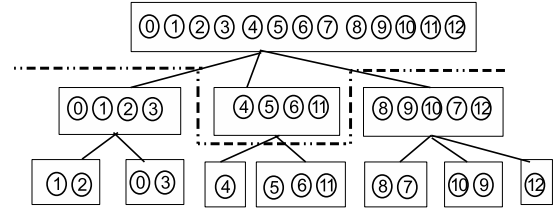}
\setlength{\abovecaptionskip}{-10pt}
\caption{Example of ModDivisive with $k=3$. A cut $C$ is shown by the dot-dashed line}
\label{fig:mod-divisive}
\vspace{-1.0em}
\end{figure}

\subsubsection{Algorithm}

Algorithm \ref{algo-moddivisive} sketches the main steps in our scheme ModDivisive. It comprises two phases: differentially private sampling a k-ary tree of depth $maxL$ which uses the privacy budget $\epsilon_1$ and finding the best cut across the tree to get a good clustering of nodes which consumes a budget $maxL.\epsilon_m$. 

The first phase (lines 1-14) begins with the creation of $eA$, the array of privacy budgets allocated to levels of the tree. We use the parameter $\lambda \geq 1$ as the common ratio to form a geometric sequence. The rationale behind the common ratio is to give higher priority to the levels near the root which have larger node sets. By sequential composition (Theorem \ref{theorem:composition}), we must have $\sum_{i}eA[i] = \epsilon_1$. All internal nodes at level $i$ do the MCMC sampling on disjoint subsets of nodes, so the parallel composition holds. Subsection \ref{subsubsec:priv} analyzes the privacy of ModDivisive in more detail. We use a queue to do a level-by-level exploration. Each dequeued node $r$'s level will be checked. If its level is not larger than $maxL$, we will run \textit{ModMCMC} (Algorithm \ref{algo-modmcmc}) on it (line 9) to get a partition $r.part$ of its nodeset $r.S$ (Fig.\ref{fig:mod-divisive}). Each subset in $r.part$ forms a child node and is pushed to the queue. The second phase (line 15) calls Algorithm \ref{algo-bestcut} to find a highly modular partition across the tree.

\begin{algorithm}
\caption{ModDivisive}
\label{algo-moddivisive}
\begin{algorithmic}[1]
	\Require graph $G$, group size $k$, privacy budget $\epsilon$, max level $maxL$, ratio $\lambda$, BestCut privacy at each level $\epsilon_m$
	\Ensure noisy partition $\tilde{C}$
	\State compute the array $eA[0..maxL-1]$ s.t. $\sum_{i}eA[i] = \epsilon_1$, $eA[i] = eA[i+1]*\lambda$ where $\epsilon_1 = \epsilon - maxL.\epsilon_m$
	\State initialize the root node with nodeset $V$
	\State $root = \text{NodeSet}(G,V,k)$
	\State $root.level = 0$
	\State queue $Q \leftarrow root$
	\While {$Q$ is not empty}
		\State $r \leftarrow Q.dequeue()$
		\If {$r.level < maxL$}
			\State $r.part = \text{ModMCMC}(G,r.S,k,eA[r.level])$	
			\For {subset $S_i$ in $r.part$}
				\State $P_i = \text{NodeSet}(G,S_i,k)$	
				\State $P_i.level = r.level + 1$	
				\State $r.child_i$ $\leftarrow P_i$ 
				\State $Q.enqueue(P_i)$
			\EndFor
		\EndIf
	\EndWhile
	\State $\tilde{C} \leftarrow \text{BestCut}(root, \epsilon_m)$ \\
	\Return $\tilde{C}$
\end{algorithmic}
\end{algorithm}

\textit{\textbf{Differentially Private Nodeset Partitioning}}
Let $\mathcal{P}$ be the space of all ways $P$ to partition a nodeset $A$ to no more than $k$ disjoint subsets, the direct application of exponential mechanism needs the enumeration of $\mathcal{P}$. 
The probability of a partition $P$ being sampled is 
\begin{equation}
\frac{\exp(\frac{\epsilon_p}{2\Delta Q} Q(P,G))} {\sum_{P'\in \mathcal{P}} \exp(\frac{\epsilon_p}{2\Delta Q} Q(P',G))}
\end{equation}
However, $|\mathcal{P}| = \sum_{i=1}^{k} S(|A|,i)$ where $S(|A|,i)$ is the Stirling number of the second kind \cite{graham1994concrete}, $S(n,k) \approx \frac{k^n}{k!}$. This sum is exponential in $|A|$, so enumerating $\mathcal{P}$ is computationally infeasible. Fortunately, MCMC can help us simulate the exponential mechanism by a sequence of local transitions in $\mathcal{P}$.

The space $\mathcal{P}$ is connected. It is straightforward to verify that the transitions performed in line 3 of ModMCMC are \textit{reversible} and \textit{ergodic} (i.e. any pair of nodeset partitions can be connected by a sequence of such transitions). Hence, ModMCMC has a unique stationary distribution in equilibrium. By empirical evaluation, we observe that ModMCMC converges after $K|r.S|$ steps for $K = 50$ (see Section \ref{subsec:eval-div}).

Each node $r$ in the tree is of type \textit{NodeSet}. This struct consists of an array $r.part$ where $r.part[u] \in \{0..k-1\}$ is the group id of $u$. To make sure that ModMCMC runs in linear time, we must have a constant time computation of modularity $Q(P)$ (line 4 of Algorithm \ref{algo-modmcmc}). This can be done with two helper arrays: the number of intra-edges $r.lc[0..k-1]$ and the total degree of nodes $r.dc[0..k-1]$ in each group. The modularity $Q$ is computed in $O(k)$ (Formula \ref{eqn:modularity}) using $r.lc, r.dc$. When moving node $u$ from group $i$ to group $j$, $r.lc$ and $r.dc$ are updated accordingly by checking the neighbors of $u$ in $G$. The average degree is a constant, so the complexity of ModMCMC is linear in the number of MCMC steps.

\begin{algorithm}
\caption{ModMCMC}
\label{algo-modmcmc}
\begin{algorithmic}[1]
	\Require graph $G$, nodeset $r.S$, group size $k$, privacy budget $\epsilon_p$
	\Ensure sampled partition $r.part$
	\State initialize $r.part$ with a random partition $P_0$ of $k$ groups
	\For {each step $i$ in the Markov chain}
		\State pick a neighboring partition $P'$ of $P_{i-1}$ by randomly selecting node $u \in r.S$ and moving $u$ to another group.
		\State accept the transition and set $P_i = P'$ with probability $min(1,\frac{\exp(\frac{\epsilon_p}{2\Delta Q} Q(P',G))} {\exp(\frac{\epsilon_p}{2\Delta Q} Q(P_{i-1},G))})$
	\EndFor
	\State // until equilibrium is reached \\
	\Return a sampled partition $r.part = P_i$
\end{algorithmic}
\end{algorithm}

\textit{\textbf{Finding the Best Cut}}
Given the output k-ary tree $R$ with the root node $root$, our next step is to find the best cut across the tree. A cut $C$ is a set of nodes in $R$ that cover all nodes in $V$. As an example, a cut $C$ in Fig.\ref{fig:mod-divisive} returns the clustering [\{0,1,2,3\}, \{4\}, \{5,6,11\}, \{8,9,10,7,12\}]. Any cut has a modularity score. Our goal is to find the best cut, i.e. the cut with highest modularity, in a private manner. 

We solve this problem by a dynamic programming technique. Remember that modularity is an \textit{additive} quantity (c.f. Formula \ref{eqn:modularity}). By denoting $opt(r)$ as the optimal modularity for the subtree rooted at node $r$, the optimal value is $opt(root)$. The recurrence relation is straightforward
\begin{equation}
opt(r) = \max \{Q(r), \sum_{t \in r.children} opt(t)\} \nonumber
\end{equation}

Algorithm \ref{algo-bestcut} realizes this idea in three steps. The first step (lines 1-6) uses a queue to fill a stack $S$. The stack ensures any internal node to be considered after its child nodes. The second step (lines 7-17) solves the recurrence relation. Because all modularity values are sensitive, we add Laplace noise $\text{Laplace}(\Delta Q/\epsilon_m)$. The global sensitivity $\Delta Q = 3/m$ (see Theorem \ref{theorem:delta-Q}), so we need only a small privacy budget for each level ($\epsilon_m = 0.01$ is enough in our experiments). The noisy modularity $mod_n$ is used to decide whether the optimal modularity at node $r$ is by itself or by the sum over its children. The final step (lines 18-25) backtracks the best cut from the root node.

\begin{algorithm}
\caption{BestCut}
\label{algo-bestcut}
\begin{algorithmic}[1]
	\Require undirected graph $G$, root node $root$, privacy budget at each level $\epsilon_m$
	\Ensure best cut $C$
	\State stack $S \leftarrow \varnothing$, queue $Q \leftarrow \text{root}$
	\While {$Q$ is not empty}
		\State $r \leftarrow Q.dequeue()$
		\State $S.push(r)$
		\For {child node $r_i$ in $r.children$}
			\State $Q.enqueue(r_i)$
		\EndFor
	\EndWhile
	\State dictionary $sol \leftarrow \varnothing$
	\While {$S$ is not empty}
		\State $r \leftarrow S.pop()$, $r.mod_n = r.mod + \text{Laplace}(\Delta Q/\epsilon_m)$
		\If {$r$ is a leaf node}			
			\State $sol.put(r.id, (\text{val}=r.mod_n, \text{self}=True))$
		\Else
			\State $s_m = \sum_{r_i \in r.children} sol[r_i.id].mod_n$
			\If {$r.mod_n < s_m$}
				\State $sol.put(r.id, (\text{val}=s_m, \text{self}=False))$
			\Else
				\State $sol.put(r.id, (\text{val}=r.mod_n, \text{self}=True))$
			\EndIf
		\EndIf
	\EndWhile
	\State list $C \leftarrow \varnothing$, queue $Q \leftarrow \text{root}$
	\While {$Q$ is not empty}
		\State $r \leftarrow Q.dequeue()$
		\If {$sol[r.id].\text{self} == True$}
		\State $C = C \cup \{r\}$
		\Else
			\For {child node $r_i$ in $r.children$}
				\State $Q.enqueue(r_i)$
			\EndFor
		\EndIf
	\EndWhile	
	\Return $C$
\end{algorithmic}
\end{algorithm}

\subsubsection{Complexity}
ModDivisive creates a k-ary tree of height $maxL$. At each node $r$ of the tree other than the leaf nodes, ModMCMC is run once. The run time of ModMCMC is $O(K*|r.S|)$ thanks to the constant time for updating the modularity (line 4 of ModMCMC). Because the union of nodesets at one level is $V$, the total runtime is $O(K*|V|*maxL)$. BestCut only incurs a sublinear runtime because the size of tree is always much smaller than $|V|$. The following theorem states this result

\begin{theorem}
The time complexity of ModDivisive is linear in the number of nodes $n$, the maximum level $maxL$ and the burn-in factor $K$.
\end{theorem}

\subsubsection{Privacy Analysis}
\label{subsubsec:priv}

We show that ModMCMC satisfies differential privacy. The goal of MCMC is to draw a random sample from the desired distribution. Similarly, exponential mechanism is also a method to sample an output $x \in X$ from the target distribution with probability proportional to $\exp(\frac{\epsilon u(x)}{2\Delta u})$ where $u(x)$ is the score function ($x$ with higher score has bigger chance to be sampled) and $\Delta u$ is its sensitivity. The idea of using MCMC to realize exponential mechanism is first proposed in \cite{chaudhuri2013near} and applied to $\epsilon$-DP graph release in \cite{xiao2014differentially}.

In our ModMCMC, the modularity $Q(P,G)$ is used directly as the score function. We need to quantify the global sensitivity of $Q$. From Section \ref{sec:diffpriv}, we have the following definition

\begin{definition} (Global Sensitivity $\Delta Q$)
\label{def:sensitivity}
\begin{equation}
\Delta Q = \max_{P,G,G'}|Q(P,G) - Q(P,G')|
\end{equation}
\end{definition}

We prove that $\Delta Q = O(1/m)$ in the following theorem
\begin{theorem}
\label{theorem:delta-Q}
The global sensitivity of modularity, $\Delta Q$, is smaller than $\frac{3}{m}$		
\end{theorem}
\begin{proof} 
\label{proof:delta-Q}
Given the graph $G$ and a partition $P$ of a nodeset $V_p \subseteq V$ (for any node of the k-ary tree other than the root node, its nodeset $V_p$ is a strict subset of $V$), the neighboring graph $G'$ has $E_{G'} = E_G \cup e$. We have two cases

\textit{Case 1.} The new edge $e$ is an intra-edge within the community $s$.
The modularity $Q(P,G)$ is $\sum_{c}^{k} (\frac{l_c}{m} - \frac{d_c^2}{4m^2})$. The modularity $Q(P,G')$ is $\sum_{c\neq s}^{k} (\frac{l_c}{m+1} - \frac{d_c^2}{4(m+1)^2}) + (\frac{l_s+1}{m+1} - \frac{(d_s+2)^2}{4(m+1)^2})$.

The difference $d_1 = Q(P,G') - Q(P,G) = \frac{1}{m+1} -\frac{1}{m(m+1)} \sum_{c}^{k}l_c + \frac{2m+1}{4m^2(m+1)^2} \sum_{c}^{k}d_c^2 - \frac{d_s+1}{(m+1)^2}$

Because $\Delta Q$ is the absolute value of $d_1$, we consider the most positive and the most negative values of $d_1$. Remember that $\sum_{c}^{k}d_c \leq 2m$, so the positive bound $d_1 < \frac{1}{m+1} + \frac{(2m+1)4m^2}{4m^2(m+1)^2} < \frac{3}{m+1}$. For the negative bound, we use the constraints $\sum_{c}^{k}l_c \leq m$ and $d_s \leq 2m$, so $d_1 > \frac{1}{m+1} - \frac{m}{m(m+1)} - \frac{2m+1}{(m+1)^2} > \frac{-2}{m+1}$.
As a result, $\Delta Q = |d_1| < \frac{3}{m+1} < \frac{3}{m}$.

\vspace{5mm} 
\textit{Case 2.} The new edge $e$ is an inter-edge between the communities $s$ and $t$.
Similarly, we have $Q(P,G) = \sum_{c}^{k} (\frac{l_c}{m} - \frac{d_c^2}{4m^2})$ while $Q(P,G') = \sum_{c\neq s,t}^{k} (\frac{l_c}{m+1} - \frac{d_c^2}{4(m+1)^2}) + (\frac{l_s}{m+1} - \frac{(d_s+1)^2}{4(m+1)^2}) + (\frac{l_t}{m+1} - \frac{(d_t+1)^2}{4(m+1)^2})$.

The difference $d_2 = Q(P,G') - Q(P,G) = -\frac{1}{m(m+1)} \sum_{c}^{k}l_c + \frac{2m+1}{4m^2(m+1)^2} \sum_{c}^{k}d_c^2 - \frac{2d_s+2d_t+2}{4(m+1)^2}$. Again, we consider the most positive and the most negative values of $d_2$, using the constraint $\sum_{c}^{k}d_c \leq 2m$, the positive bound $d_2 < \frac{(2m+1)4m^2}{4m^2(m+1)^2} < \frac{2}{m+1}$. For the negative bound, we use the constraints $\sum_{c}^{k}l_c \leq m$ and $d_s+d_t \leq 2m$, so $d_2 > - \frac{m}{m(m+1)} - \frac{4m+2}{4(m+1)^2} > \frac{-2}{m+1}$. As a result, $\Delta Q = |d_2| < \frac{2}{m+1} < \frac{2}{m}$.

To recap, in both cases $\Delta Q < \frac{3}{m}$.
\end{proof}

\subsection{A Variant of HRG-MCMC}
\label{subsec:hrg-fixed}
Similar to DER and EdgeFlip, HRG-MCMC \cite{xiao2014differentially} runs in quadratic time due to its costly MCMC steps. Each MCMC step in HRG-MCMC takes $O(n)$ to update the tree. To make it runnable on million-scale graphs, we describe briefly a variant of HRG-MCMC called \textit{HRG-Fixed} (see Fig.\ref{fig:hrg-fixed}). Instead of exploring the whole space of HRG trees as in HRG-MCMC, HRG-Fixed selects a fixed binary tree beforehand. We choose a balanced tree in our HRG-Fixed. Then HRG-Fixed realizes the exponential mechanism by sampling a permutation of $n$ leaf nodes (note that each leaf node represents a graph node). The next permutation is constructed from the current one by randomly choosing a pair of nodes and swap them. The bottom of Fig.\ref{fig:hrg-fixed} illustrates the swap of two nodes $a$ and $d$. The log-likelihood $L$ still has the sensitivity $\Delta L = 2\ln n$ as in HRG-MCMC \cite{xiao2014differentially}. Each sampling (MCMC) operation is designed to run in O($\bar{d}.\log n$). By running $K.n$ MCMC steps, the total runtime of HRG-Fixed is $O(K.m.\log n)$, feasible on large graphs. The burn-in factor $K$ of HRG-Fixed is set to 1,000 in our evaluation.

\begin{figure}
\centering
\includegraphics[height=1.8in]{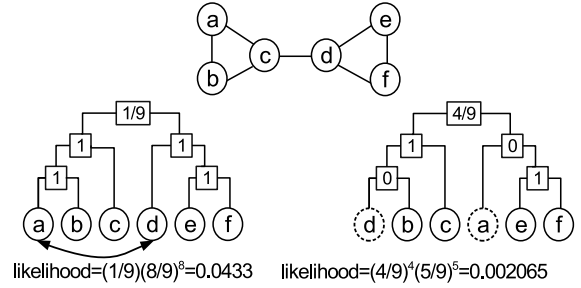}
\setlength{\abovecaptionskip}{-10pt}
\caption{HRG-Fixed}
\vspace{-1.0em}
\label{fig:hrg-fixed}
\end{figure}

\section{Experiments and Results}
\label{sec:eval}

In this section, our evaluation aims to compare the performance of the competitors by clustering quality and efficiency. The clustering quality is measured by the modularity $Q$ and the average $F_1$-score in which the modularity is the most important metric as we aim at highly modular clusterings. The efficiency is measured by the running time. All algorithms are implemented in Java and run on a desktop PC with $Intel^{\circledR}$ Core i7-4770@ 3.4Ghz, 16GB memory. 

Two medium-sized and three large real graphs are used in our experiments \footnote{http://snap.stanford.edu/data/index.html}. \texttt{as20graph} is the graph of routers comprising the internet. \texttt{ca-AstroPh} and \texttt{dblp} are co-authorship networks where two authors are connected if they publish at least one paper together. \texttt{amazon} is a product co-purchasing network where the graph contains an undirected edge from $i$ to $j$ if a product $i$ is frequently co-purchased with product $j$. \texttt{youtube} is a video-sharing web site that includes a social network. Table \ref{tab:dataset} shows the characteristics of the graphs. The columns Com(munities) and Mod(ularity) are the output of Louvain method. The number of samples in each test case is 20.

\begin{table}[htb]
\centering
\caption{Characteristics of the test graphs} \label{tab:dataset}
\begin{tabular}{l|r|r|r|r}
 & Nodes & Edges & Com & Mod  \\ 
\hline
 as20graph & 6,474 & 12,572 & 30 & 0.623 \\ 
\hline
 ca-AstroPh & 17,903 & 196,972 & 37 & 0.624 \\ 
\hline
 amazon & 334,863 & 925,872 & 257 & 0.926 \\ 
\hline
 dblp & 317,080 & 1,049,866 & 375 & 0.818 \\ 
\hline
 youtube & 1,134,890 & 2,987,624 & 13,485 & 0.710 \\ 
\end{tabular}
\end{table}

The schemes are abbreviated as 1K-series (1K), EdgeFlip (EF), Top-m-Filter (TmF), DER, LouvainDP (LDP), ModDivisive (MD), HRG-MCMC and HRG-Fixed.

\subsection{Quality Metrics}
Apart from modularity $Q$, we use $\bar{F_1}$, the \textit{average $F_1$-score}, following the benchmarks in \cite{yang2012defining}. The $F_1$ score of a set $A$ with respect to a set $B$ is defined as the harmonic mean $H$ of the precision and the recall of $A$ against $B$. We define $prec(A,B) = \frac{|A \cap B|}{|A|}$, $recall(A,B) = \frac{|A \cap B|}{|B|}$
\begin{equation}
F_1(A,B) = \frac{2.prec(A,B).recall(A,B)}{prec(A,B) + recall(A,B)}  \nonumber
\end{equation}

Then the average $F_1$ score of two sets of communities $C$ and $C'$ is defined as
\begin{multline*}
F_1(A,C) = \max_{i} F_1(A,c_i), \;\; c_i \in C = \{c_1,..,c_n\} \\   \nonumber
\bar{F_1}(C,C') = \frac{1}{2|C|} \sum_{c_i\in C} F_1(c_i,C') + \frac{1}{2|C'|} \sum_{c_i\in C'} F_1(c_i,C) \nonumber
\end{multline*}

We choose the output clustering of Louvain method as the ground truth for two reasons. First, the evaluation on the real ground truth is already done in \cite{prat2014high} and Louvain method is proven to provide high quality communities. Second, the real ground truth is a set of  \textit{overlap} communities whereas the schemes in this paper output only \textit{non-overlap} communities. The chosen values of $\epsilon$ are \{0.1$\ln n$,0.2$\ln n$,0.3$\ln n$,0.4$\ln n$,0.5$\ln n$\}.

\subsection{Performance of LouvainDP}
We test LouvainDP for the group size $k \in \{4,8,16,32,64\}$. The results on \texttt{youtube} are displayed in Fig. \ref{fig:louvaindp}. We observe a clear separation of two groups $k=4,8$ and $k=16,32,64$. As $k$ increases, the modularity increases faster but also saturates sooner. Similar separations appear in avg.F1Score and the number of communities. Non-trivial modularity scores in several settings of $(k,\epsilon)$ indicate that randomly grouping of nodes and high-pass filtering superedges do not destroy all community structure of the original graph.

At $\epsilon = 0.5\ln n$ and $k = 4,8$, the total edge weight in $G_1$ is very low ( $< 0.05m$), so many supernodes of $G_1$ are disconnected and Louvain method outputs a large number of communities (Fig. \ref{fig:louvaindp-youtube-com}). The reason is that the threshold $\theta$ in LouvainDP is an integral value, so causing abnormal leaps in the total edge weight of $G_1$. We pick $k=8,64$ for the comparative evaluation (Section \ref{subsec:eval-comp}).

\begin{figure*}
 	\centering
         \begin{subfigure}[b]{0.32\textwidth}
                 \centering
                 \includegraphics[height=1.6in]{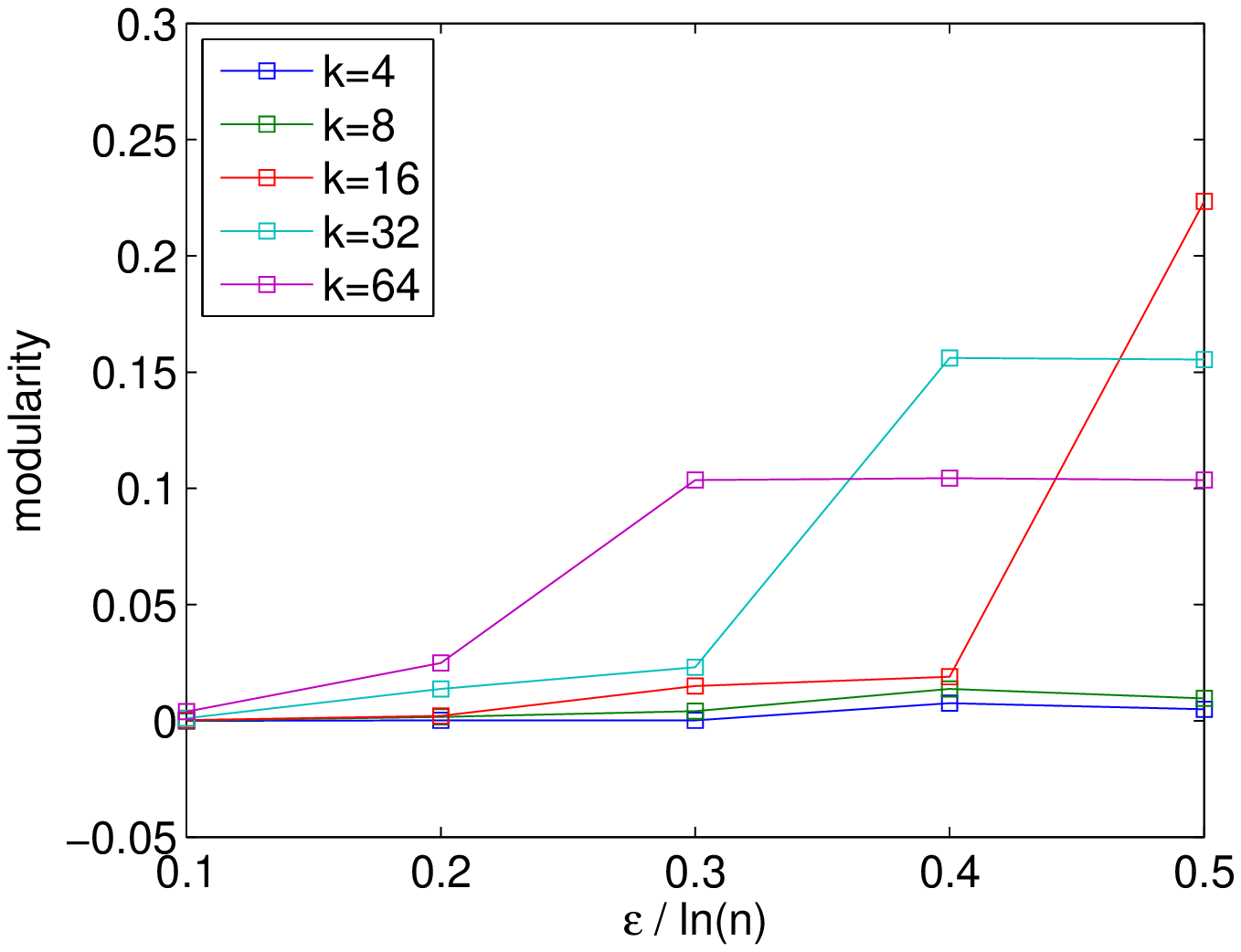}
                 \caption{}
                 \label{fig:louvaindp-youtube-mod}
         \end{subfigure}
         \hfill
         \begin{subfigure}[b]{0.32\textwidth}
                 \centering
                 \includegraphics[height=1.6in]{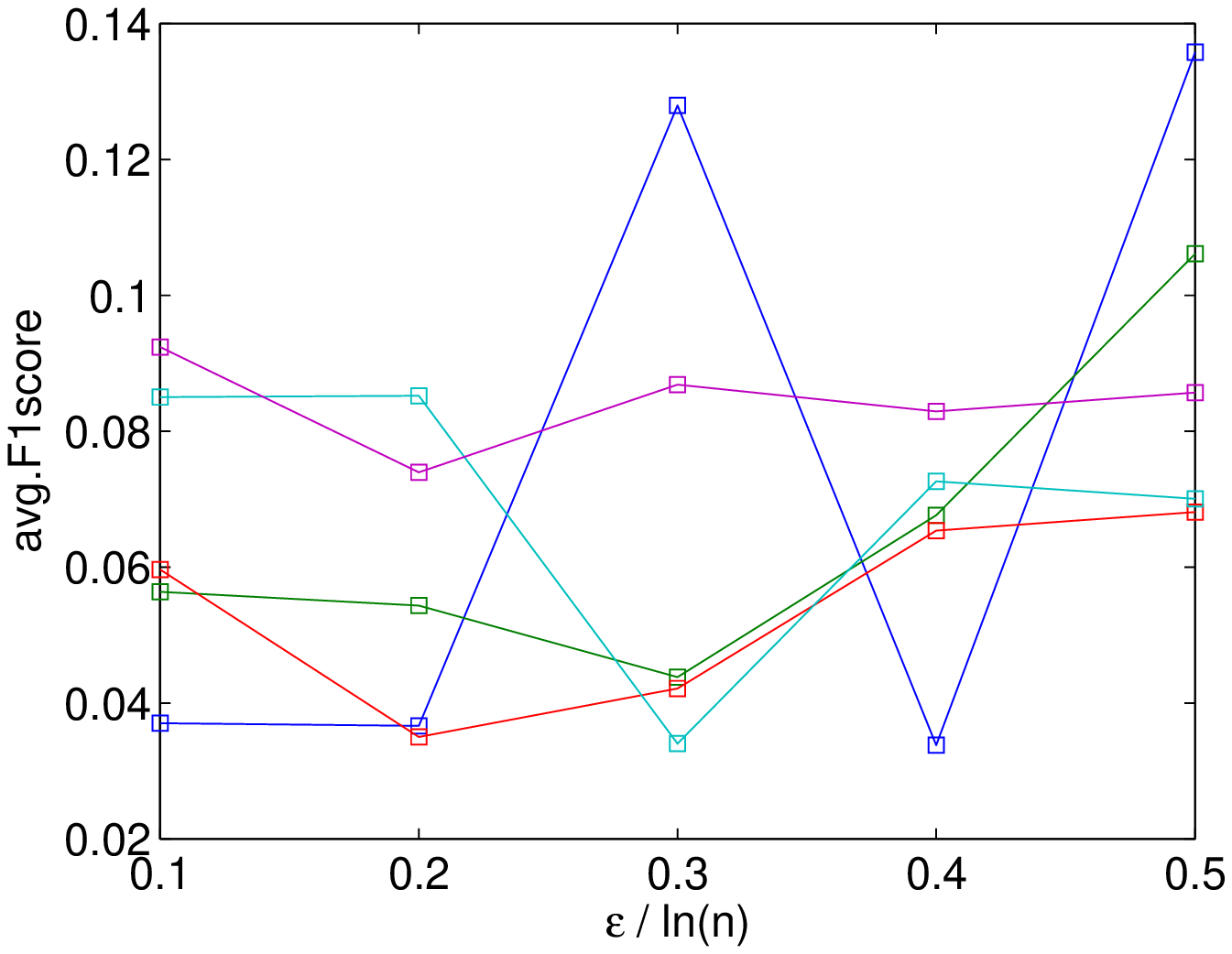}
                 \caption{}
                 \label{fig:louvaindp-youtube-f1}
         \end{subfigure}
         \hfill
         \begin{subfigure}[b]{0.32\textwidth}
                 \centering
                 \includegraphics[height=1.6in]{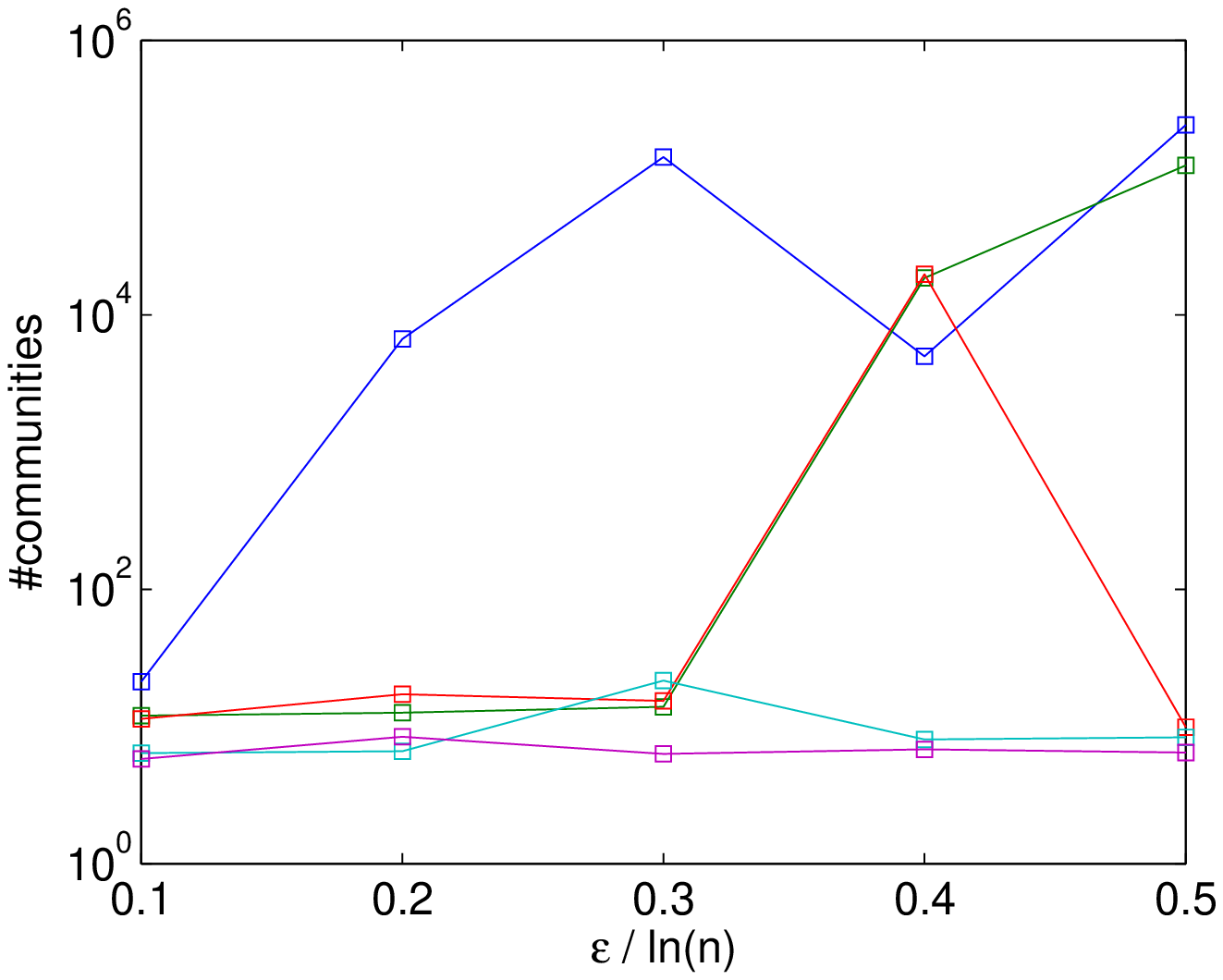}
                 \caption{}
                 \label{fig:louvaindp-youtube-com}
          \end{subfigure}        	
     \caption{LouvainDP on \texttt{youtube} (0.5$\ln n$ = 7.0) }
     \label{fig:louvaindp}
 \end{figure*}

\subsection{Performance of ModDivisive}
\label{subsec:eval-div}

The effectiveness of ModDivisive is illustrated in Fig. \ref{fig:moddivisive} for graph \texttt{youtube} and $\lambda = 2.0$, $K =50$. We select six pairs of $(k,maxL)$ by the set \{(2,10),(3,7),(4,5), (5,4),(6,4),(10,3)\}. Modularity increases steadily with $\epsilon$ while it is not always the case for avg.F1Score.
The number of communities in the best cut is shown in Fig. \ref{fig:moddivisive-youtube-com}. Clearly, the small number of communities indicates that ModDivisive's best cut is not far from the root. The reason is the use of $\lambda = 2.0$, i.e. half of privacy budget is reserved to the first level, the half of the rest for the second level and so on. Lower levels receive geometrically smaller privacy budgets, so their partitions get poorer results. 

\begin{figure*}
 	\centering
         \begin{subfigure}[h]{0.32\textwidth}
                 \centering
                 \includegraphics[height=1.6in]{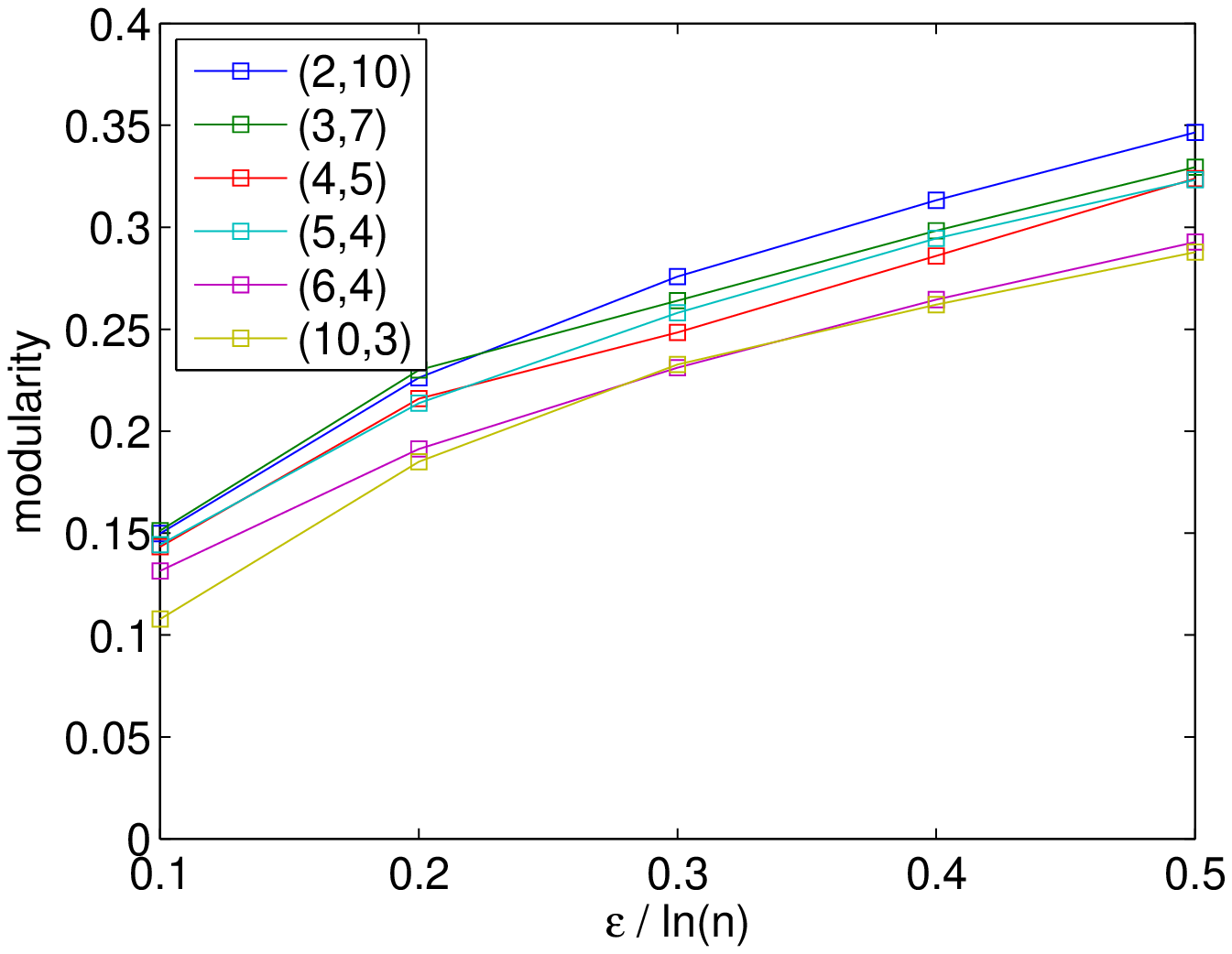}
                 \caption{}	
                 \label{fig:moddivisive-youtube-mod}
         \end{subfigure}
         \hfill
         \begin{subfigure}[h]{0.32\textwidth}
                 \centering
                 \includegraphics[height=1.6in]{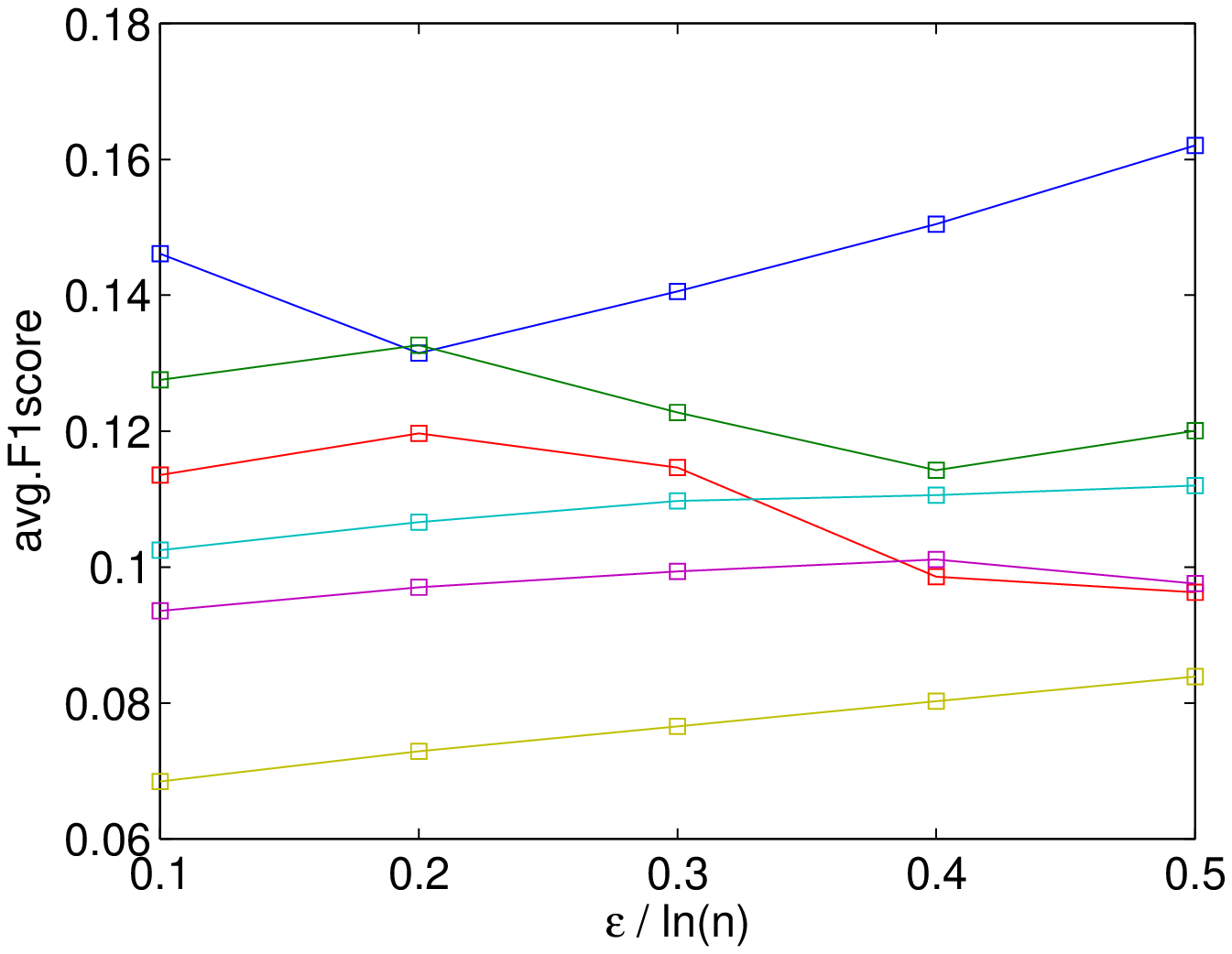}
                 \caption{}
                 \label{fig:moddivisive-youtube-f1}
         \end{subfigure}
         \hfill         
         \begin{subfigure}[h]{0.32\textwidth}
                  \centering
                  \includegraphics[height=1.6in]{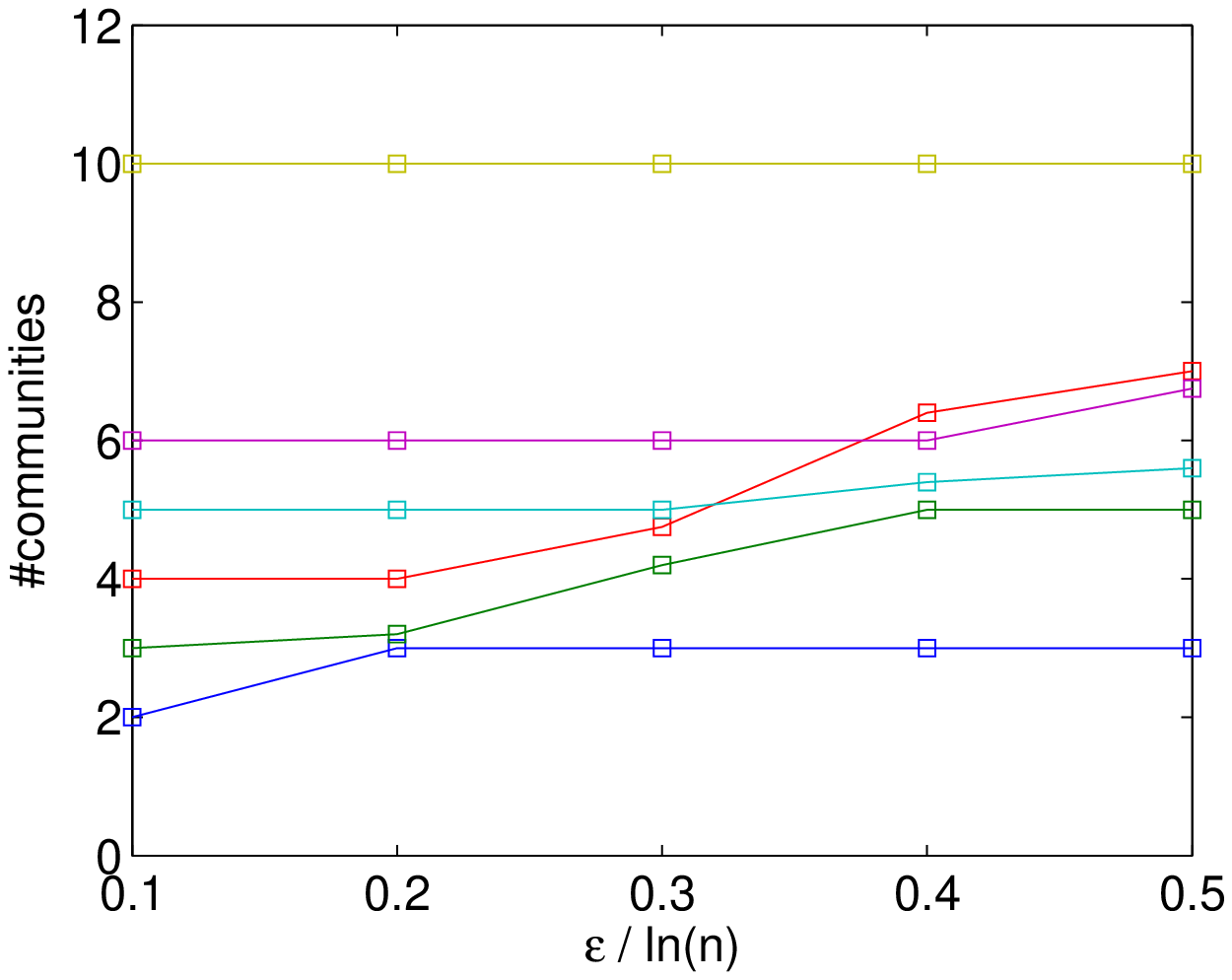}
                 \caption{}
                  \label{fig:moddivisive-youtube-com}
          \end{subfigure}        	
     \caption{ModDivisive on \texttt{youtube} with $\lambda = 2.0$, $K=50$ (0.5$\ln n$ = 7.0)}
     \label{fig:moddivisive}
\end{figure*}

 We choose $\lambda = 2.0$ to obtain a good allocation of $\epsilon$ among the levels. Fig. \ref{fig:ratio-amazon} shows the modularity for different values of $\lambda$. Note that $\lambda = 1.0$ means $\epsilon$ is equally allocated to the $maxL$ levels. By building a k-ary tree, we reduce considerably the size of the state space $\mathcal{P}$ for MCMC. As a result, we need only a small burn-in factor $K$. Looking at Fig. \ref{fig:burnfactor-amazon}, we see that larger $K = 100$ results in only tiny increase of modularity in comparison with that of $K = 50$.
 
\begin{figure}
	\centering
        \begin{subfigure}[b]{0.22\textwidth}
                \centering
                \includegraphics[height=1.4in]{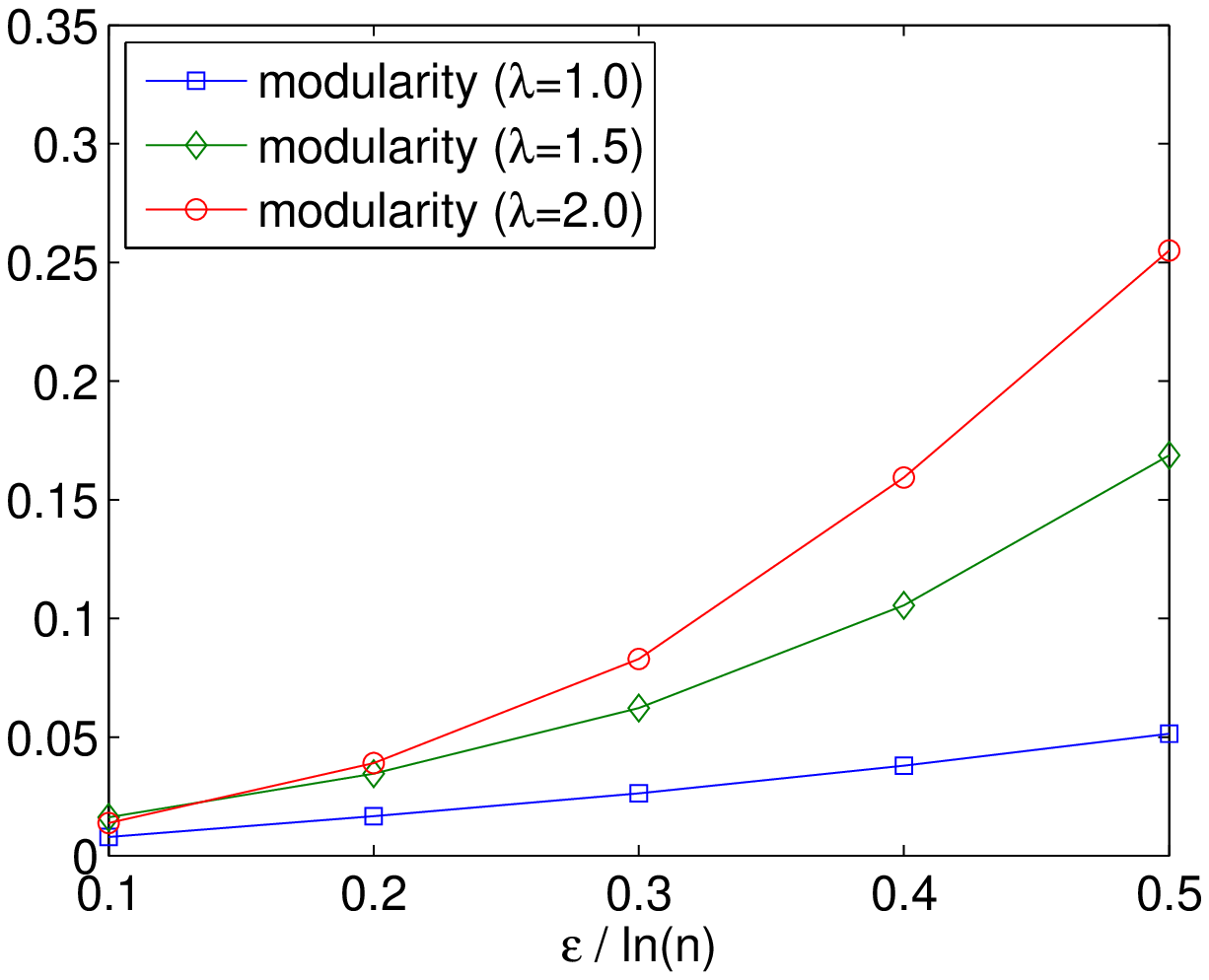}
                \setlength{\abovecaptionskip}{-10pt}
                \caption{$\lambda$}	
                \label{fig:ratio-amazon}
        \end{subfigure}
        \hfill
        \begin{subfigure}[b]{0.22\textwidth}
                \centering
                \includegraphics[height=1.4in]{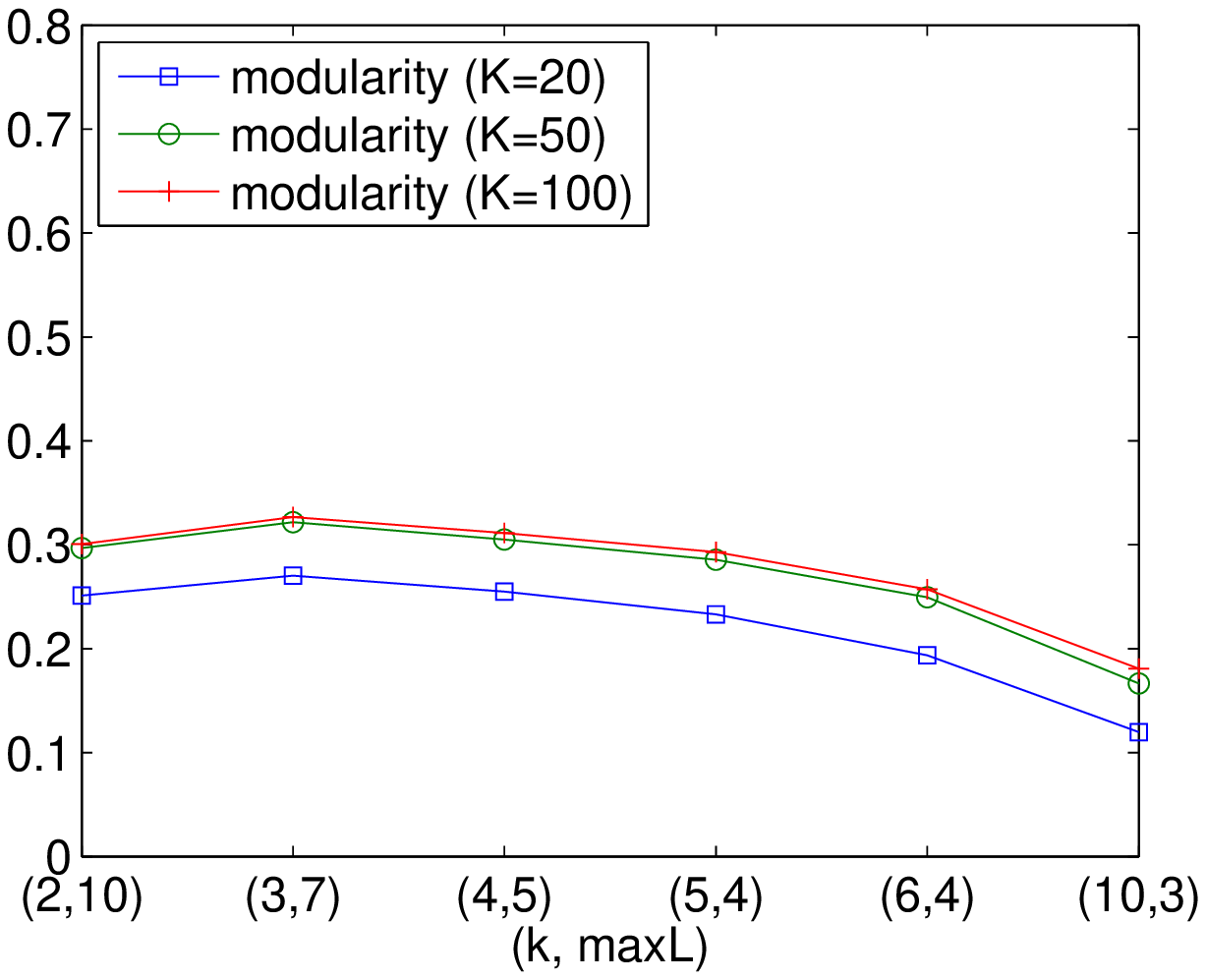}
                \setlength{\abovecaptionskip}{-10pt}
                \caption{$K$ (at $\epsilon$ = 0.5$\ln n$)}	
                \label{fig:burnfactor-amazon}
        \end{subfigure}               
    \caption{ModDivisive: modularity vs. $\lambda$ and $K$ on \texttt{amazon}}
    \label{fig:moddiv-ratio-burnfactor}
\end{figure}

\subsection{Comparative Evaluation}
\label{subsec:eval-comp}
In this section, we report a comparative evaluation of LouvainDP and ModDivisive against the competitors in Figures \ref{fig:all-as20graph}, \ref{fig:all-caAstroPh}, \ref{fig:all-amazon}, \ref{fig:all-dblp} and \ref{fig:all-youtube}. The dashed lines in subfigures \ref{fig:all-as20graph-com}, \ref{fig:all-caAstroPh-com}, \ref{fig:all-amazon-com}, \ref{fig:all-dblp-com} and \ref{fig:all-youtube-com} represent the ground-truth number of communities by Louvain method. ModDivisive performs best in most of the cases. 

On \texttt{as20graph} and \texttt{ca-AstroPh}, HRG-MCMC outputs the whole nodeset $V$ with the zero modularity while 1K-series, TmF, DER also give useless clusterings. EdgeFlip produces good quality metrics exclusively on \texttt{ca-AstroPh} while LouvainDP returns the highest modularity scores on \texttt{as20graph}. However, the inherent quadratic complexity of EdgeFlip makes Louvain method fail at $\epsilon = 0.1\ln n$ and $0.2\ln n$ for \texttt{ca-AstroPh} graph.

On the three large graphs, ModDivisive dominates the other schemes by a large margin in modularity and avg.F1Score. LouvainDP is the second best in modularity at $k=64$. 
Our proposed HRG-Fixed is consistent with $\epsilon$ and has good performance on \texttt{dblp} and \texttt{youtube}. Note that HRG-MCMC is infeasible on the three large graphs due to its quadratic complexity. Again, 1K-series, TmF and EdgeFlipShrink provide the worst quality scores with the exception of 1K-series's avg.F1Score on \texttt{youtube}.

The runtime of the linear schemes is reported in Fig. \ref{fig:runtime}. EdgeFlipShrink, 1K-series, TmF and LouvainDP benefit greatly by running Louvain method on the noisy output graph $\tilde{G}$. ModDivisive and HRG-Fixed also finish their work quickly in $O(K.n.maxL)$ and $O(K.m.\log n)$ respectively.

\begin{figure}
	\centering
        \begin{subfigure}[b]{0.22\textwidth}
                \centering
                \includegraphics[height=1.4in]{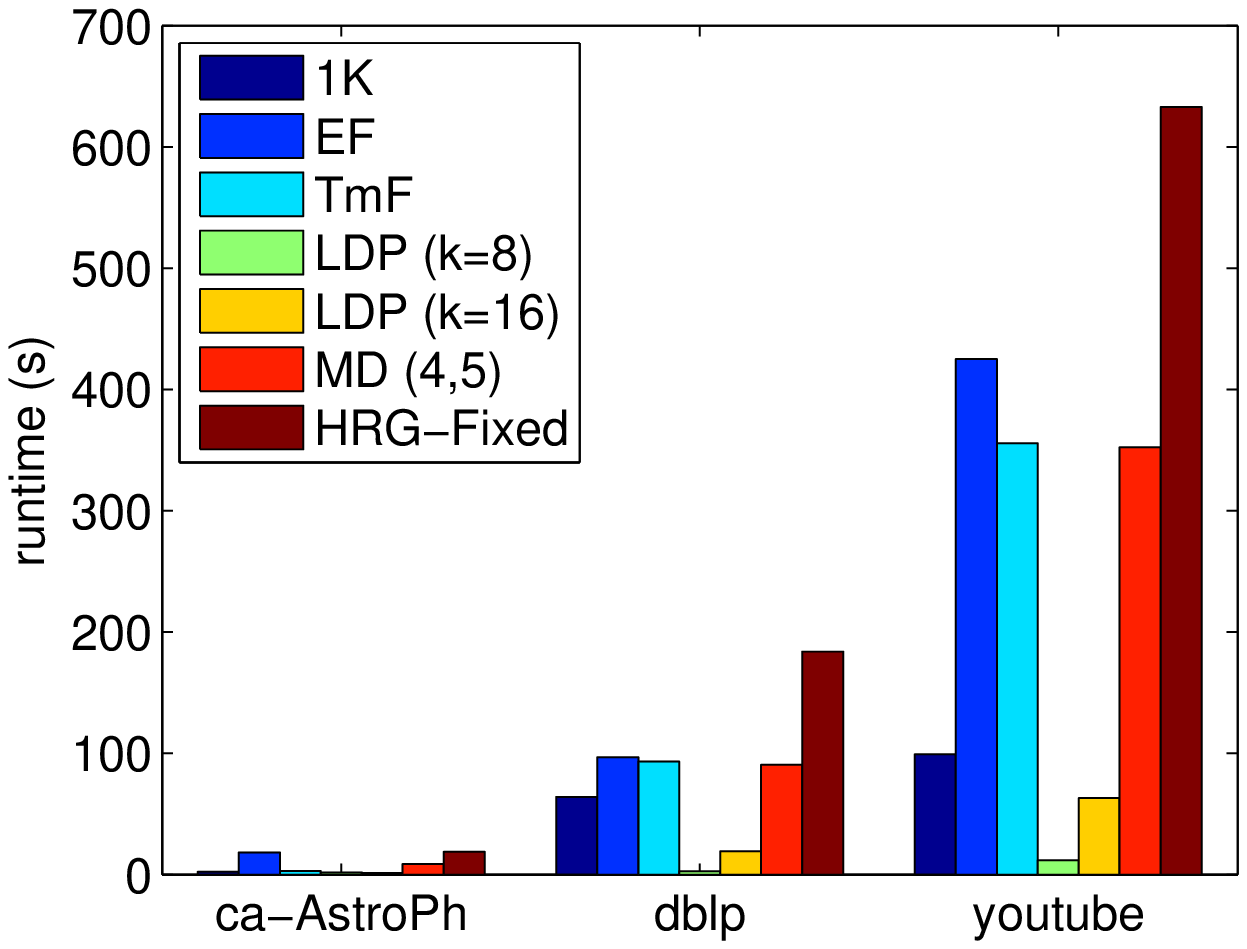}
                \setlength{\abovecaptionskip}{-10pt}
                \caption{All schemes ($\epsilon$ = 0.5$\ln n$)}	
                \label{fig:runtime-all}
        \end{subfigure}
        \hfill
        \begin{subfigure}[b]{0.22\textwidth}
                \centering
                \includegraphics[height=1.4in]{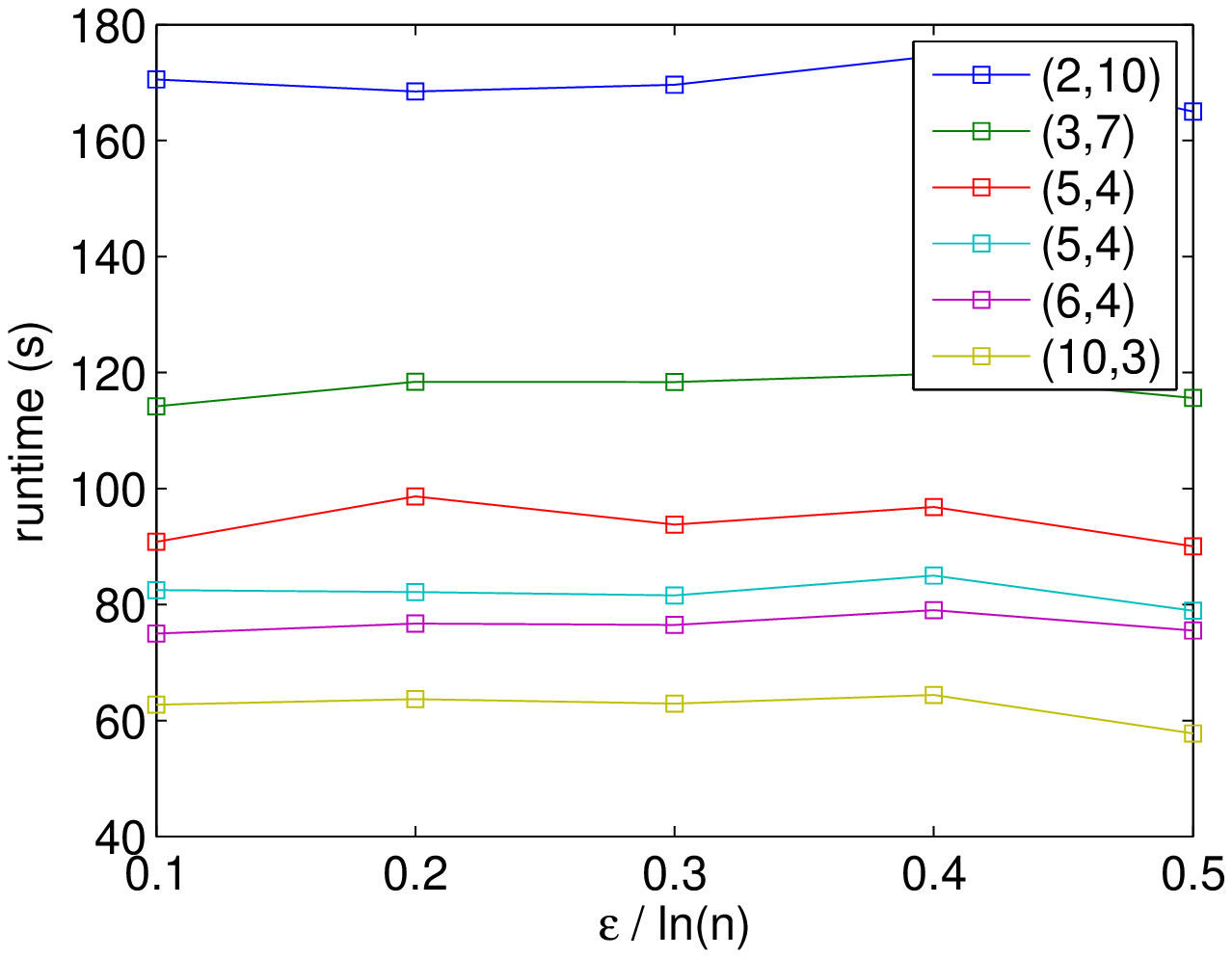}
                \setlength{\abovecaptionskip}{-10pt}
                \caption{ModDivisive on \texttt{amazon}}	
                \label{fig:runtime-moddiv-amazon}
        \end{subfigure}               
    \caption{Runtime}
    \label{fig:runtime}
\end{figure}

\section{Conclusion}
\label{sec:conclusion}
We have given a big picture of the problem $\epsilon$-DP community detection within the two categories: input and algorithm perturbation. We analyzed the major challenges of community detection under differential privacy. We explained why techniques borrowed from k-Means fail and how the difficulty of $\epsilon$-DP recommender systems enables a relaxation of privacy budget. We proposed LouvainDP and ModDivisive as the representatives of input and algorithm perturbations respectively. By conducting a comprehensive evaluation, we revealed the advantages of our methods. ModDivisive steadily gives the best modularity and avg.F1Score on large graphs while LouvainDP outperforms the remaining input perturbation competitors in certain settings. HRG-MCMC/HRG-Fixed give low modularity clusterings, indicating the limitation of the HRG model in divisive CD. The input perturbation schemes DER, EF, 1K-series and TmF hardly deliver any good node clustering except EF on the two medium-sized graphs. 

For future work, we plan to develop an $\epsilon$-DP agglomerative scheme based on Louvain method and extend our work for directed graphs and overlapping community detection under differential privacy.

\begin{figure*}
  	\centering
          \begin{subfigure}[t]{0.32\textwidth}
                  \centering
                  \includegraphics[height=1.6in]{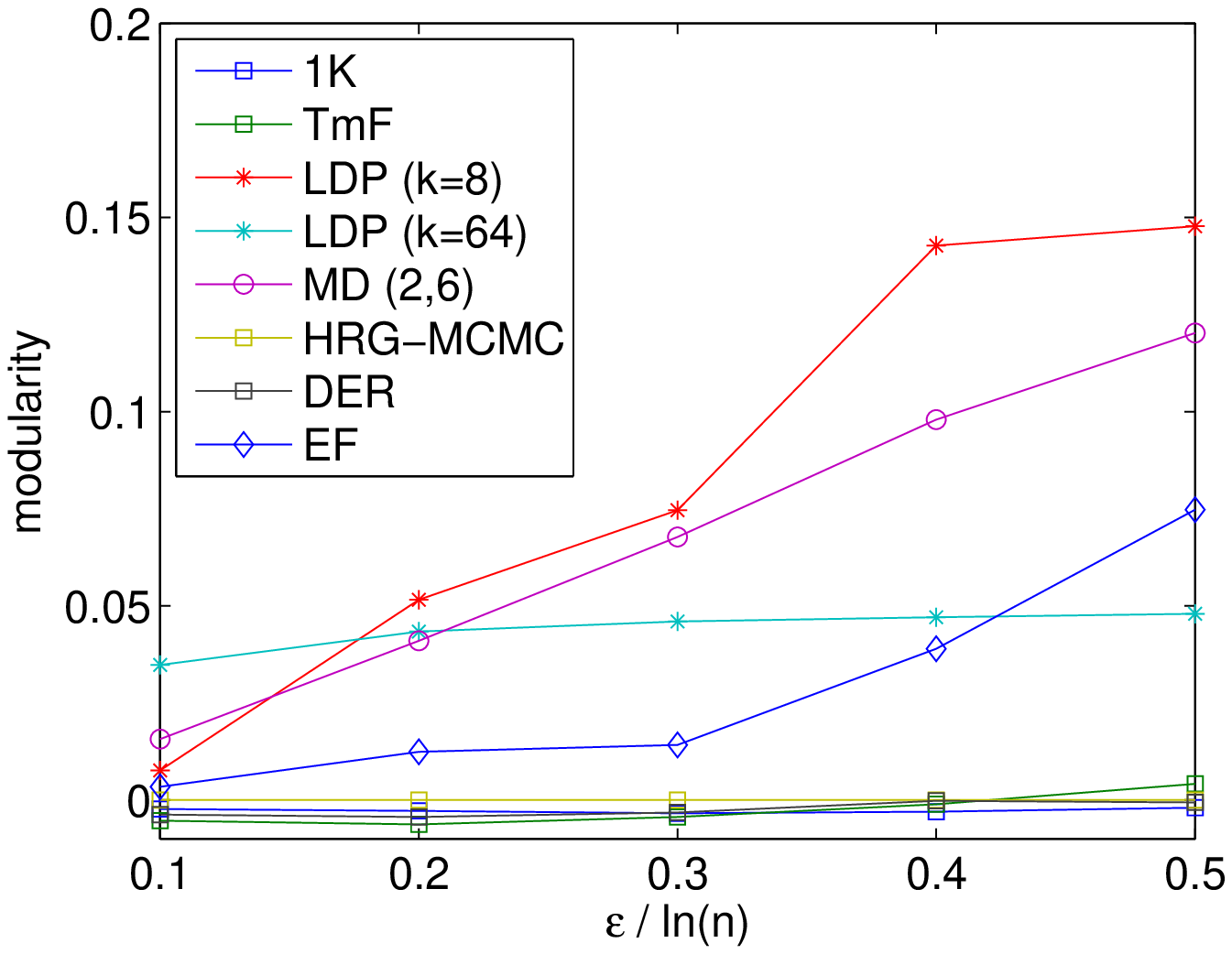}
                  \caption{}	
                  \label{fig:all-as20graph-mod}
          \end{subfigure}
          \hfill
          \begin{subfigure}[t]{0.32\textwidth}
                  \centering
                  \includegraphics[height=1.6in]{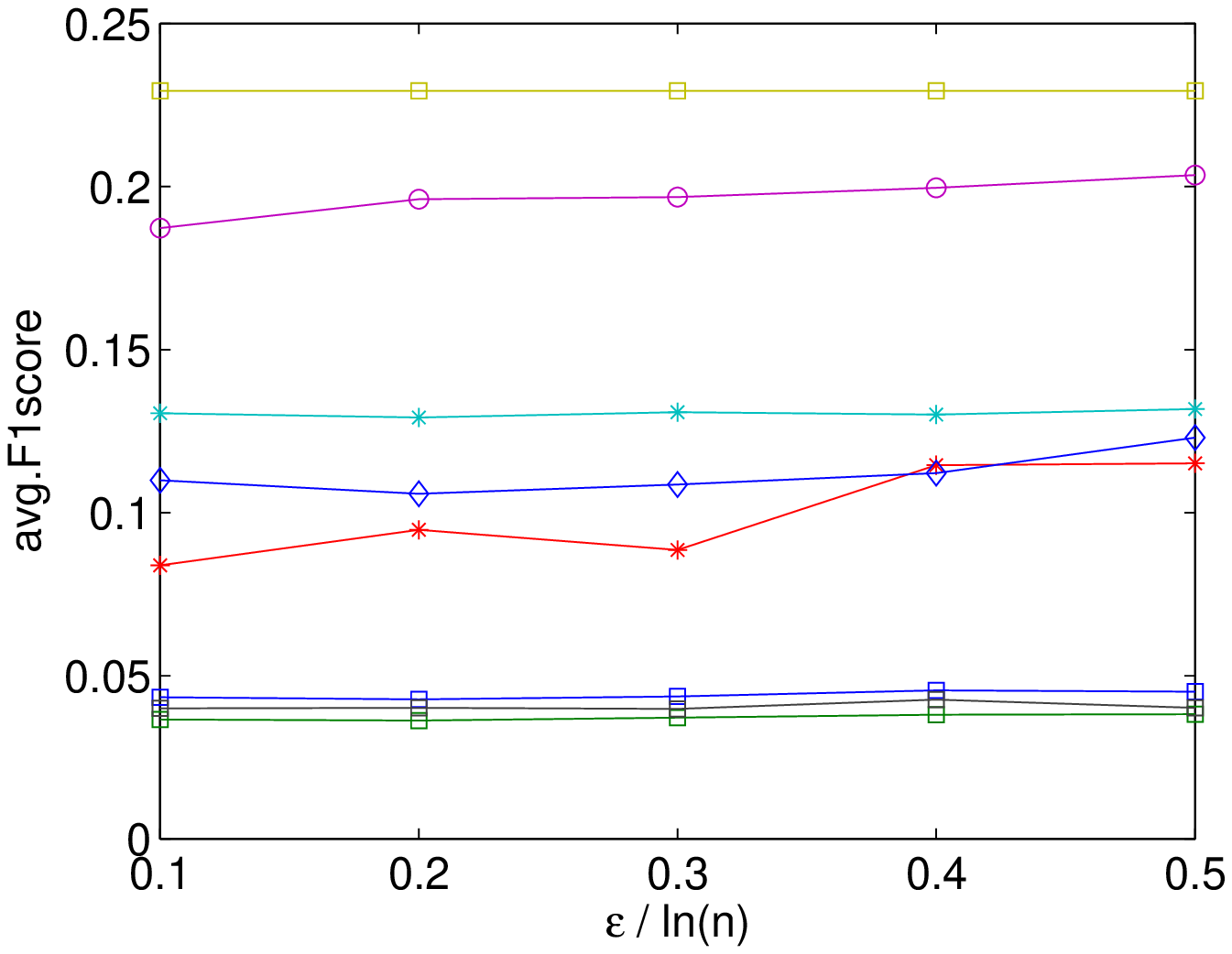}
                  \caption{}
                  \label{fig:all-as20graph-f1}
          \end{subfigure}
          \hfill         
          \begin{subfigure}[t]{0.32\textwidth}
                   \centering
                   \includegraphics[height=1.6in]{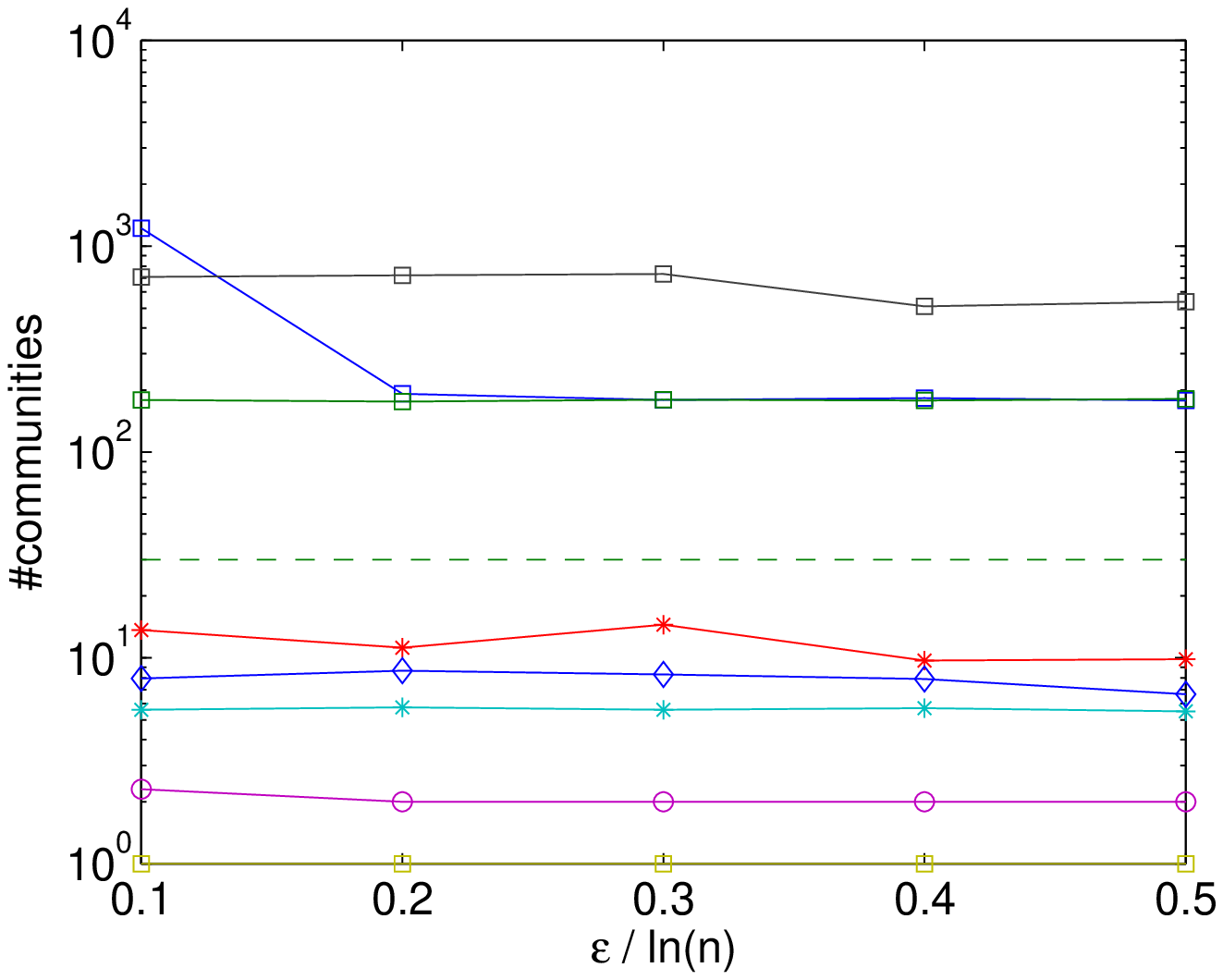}
                  \caption{}
                   \label{fig:all-as20graph-com}
           \end{subfigure}        	
      \caption{Quality metrics and the number of communities (\texttt{as20graph}) (0.5$\ln n$ = 4.4)}
      \label{fig:all-as20graph}
  \end{figure*}
  
\begin{figure*}
 	\centering
         \begin{subfigure}[t]{0.32\textwidth}
                 \centering
                 \includegraphics[height=1.6in]{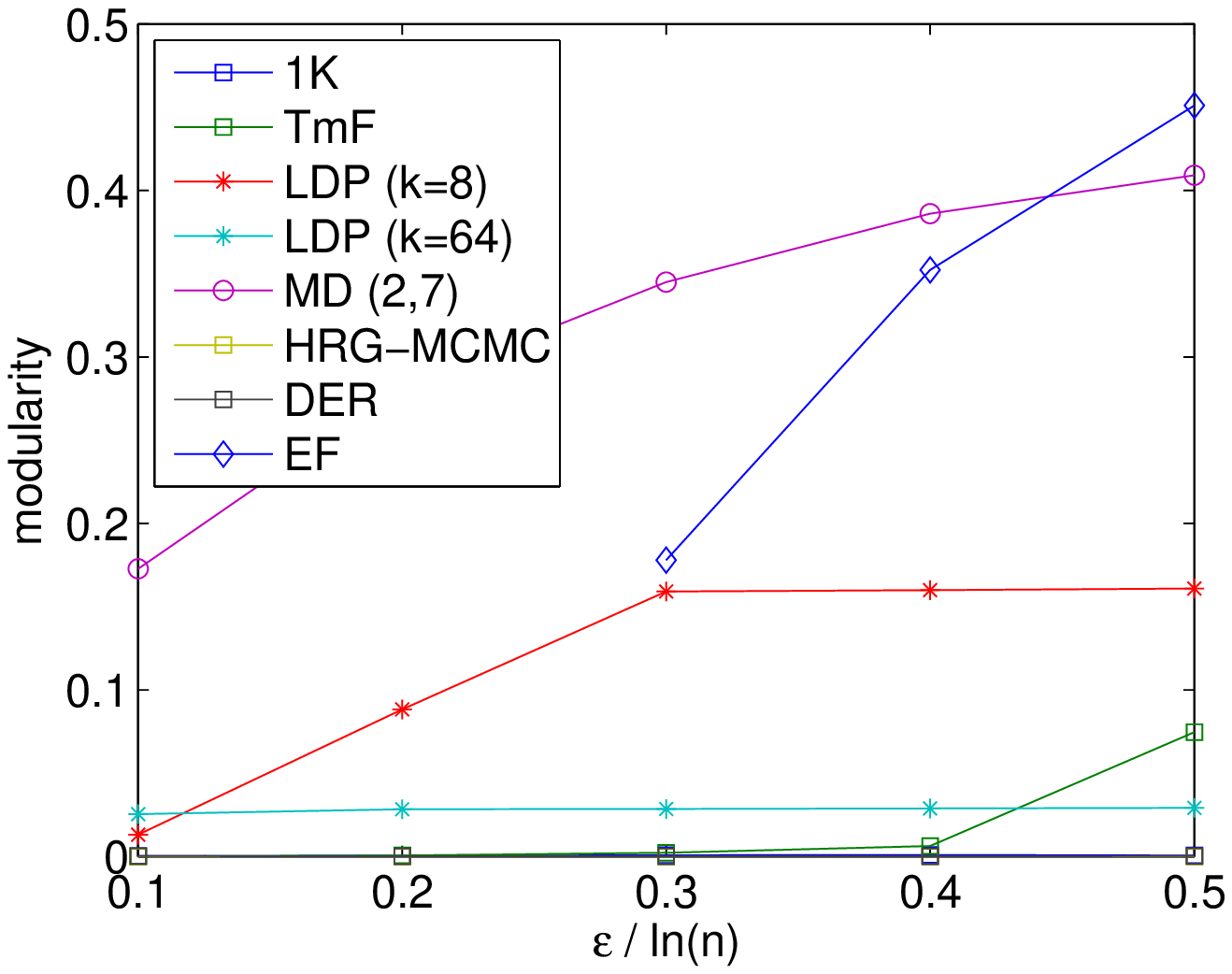}
                 \caption{}	
                 \label{fig:all-caAstroPh-mod}
         \end{subfigure}
         \hfill
         \begin{subfigure}[t]{0.32\textwidth}
                 \centering
                 \includegraphics[height=1.6in]{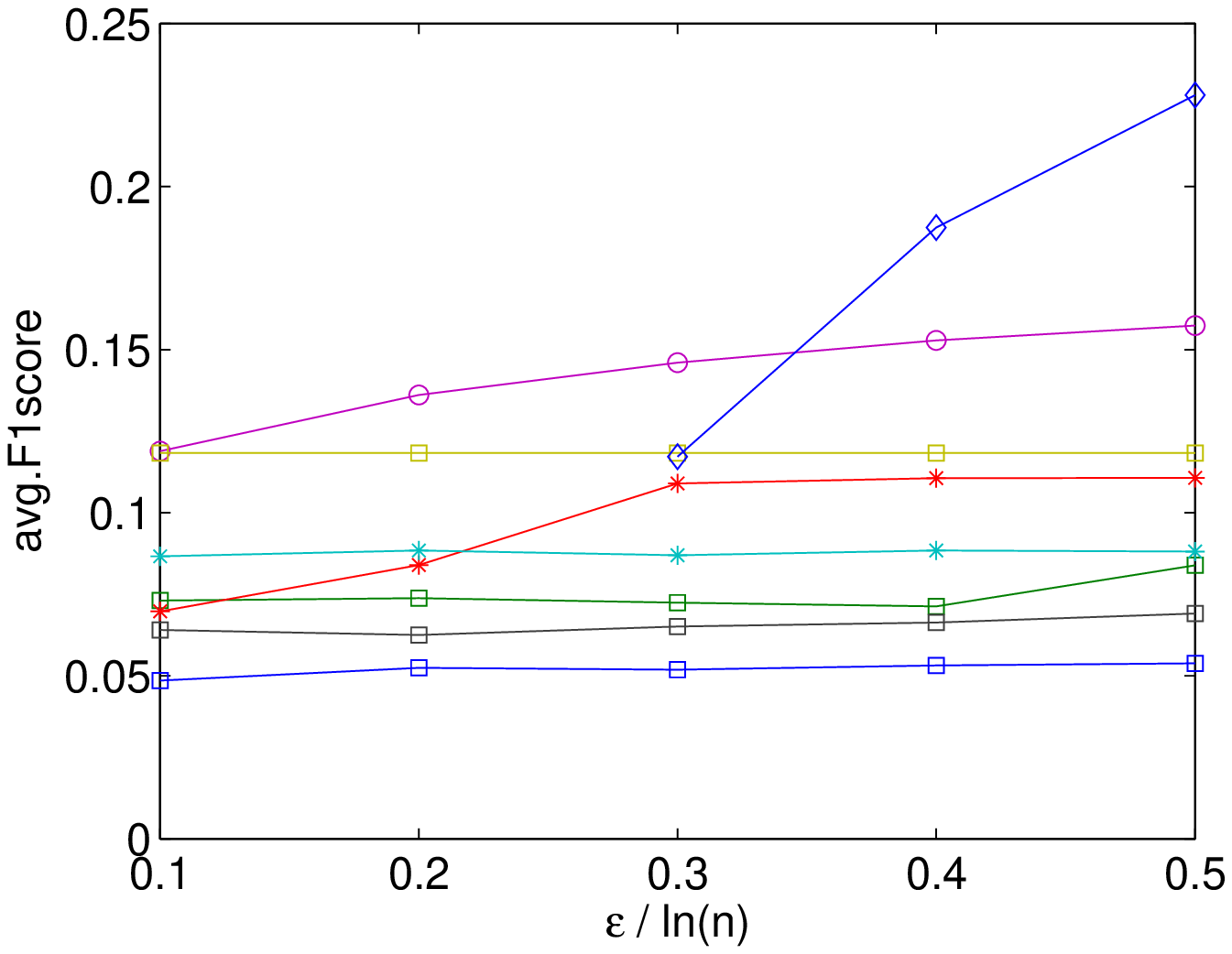}
                 \caption{}
                 \label{fig:all-caAstroPh-f1}
         \end{subfigure}
         \hfill         
         \begin{subfigure}[t]{0.32\textwidth}
                  \centering
                  \includegraphics[height=1.6in]{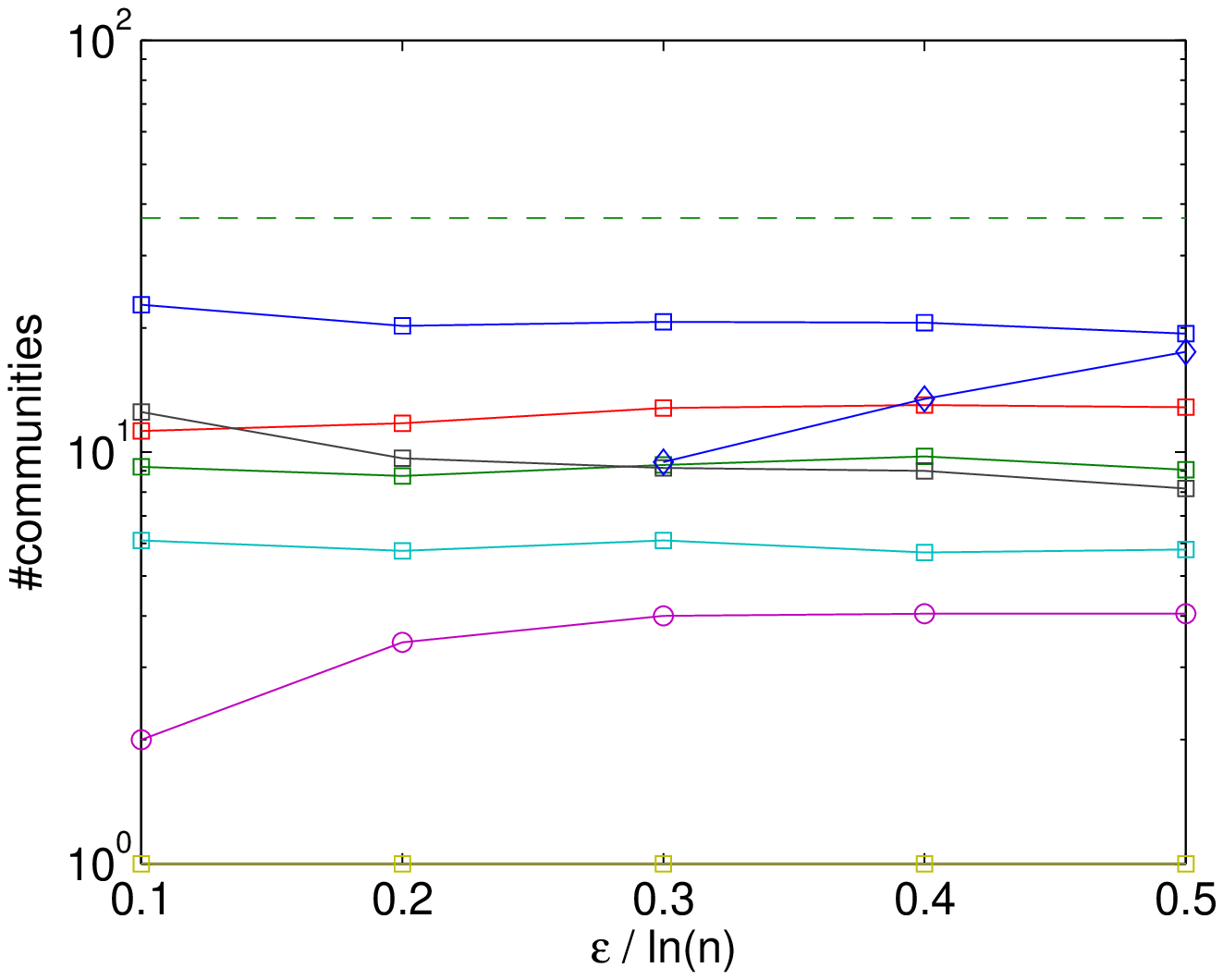}
                 \caption{}
                  \label{fig:all-caAstroPh-com}
          \end{subfigure}        	
     \caption{Quality metrics and the number of communities (\texttt{ca-AstroPh}) (0.5$\ln n$ = 4.9)}
     \label{fig:all-caAstroPh}
 \end{figure*}

\begin{figure*}
 	\centering
         \begin{subfigure}[t]{0.32\textwidth}
                 \centering
                 \includegraphics[height=1.6in]{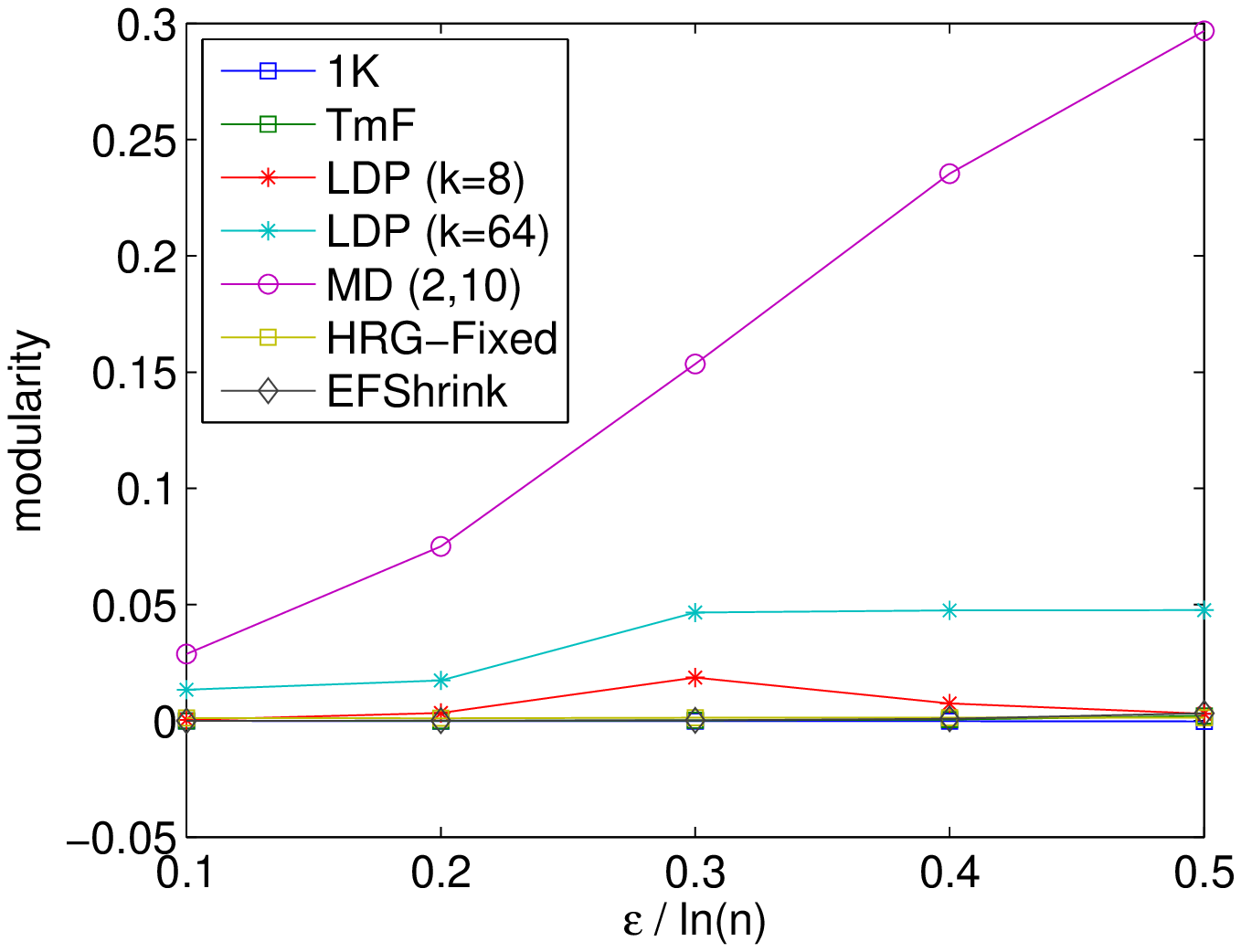}
                 \caption{}	
                 \label{fig:all-amazon-mod}
         \end{subfigure}
         \hfill
         \begin{subfigure}[t]{0.32\textwidth}
                 \centering
                 \includegraphics[height=1.6in]{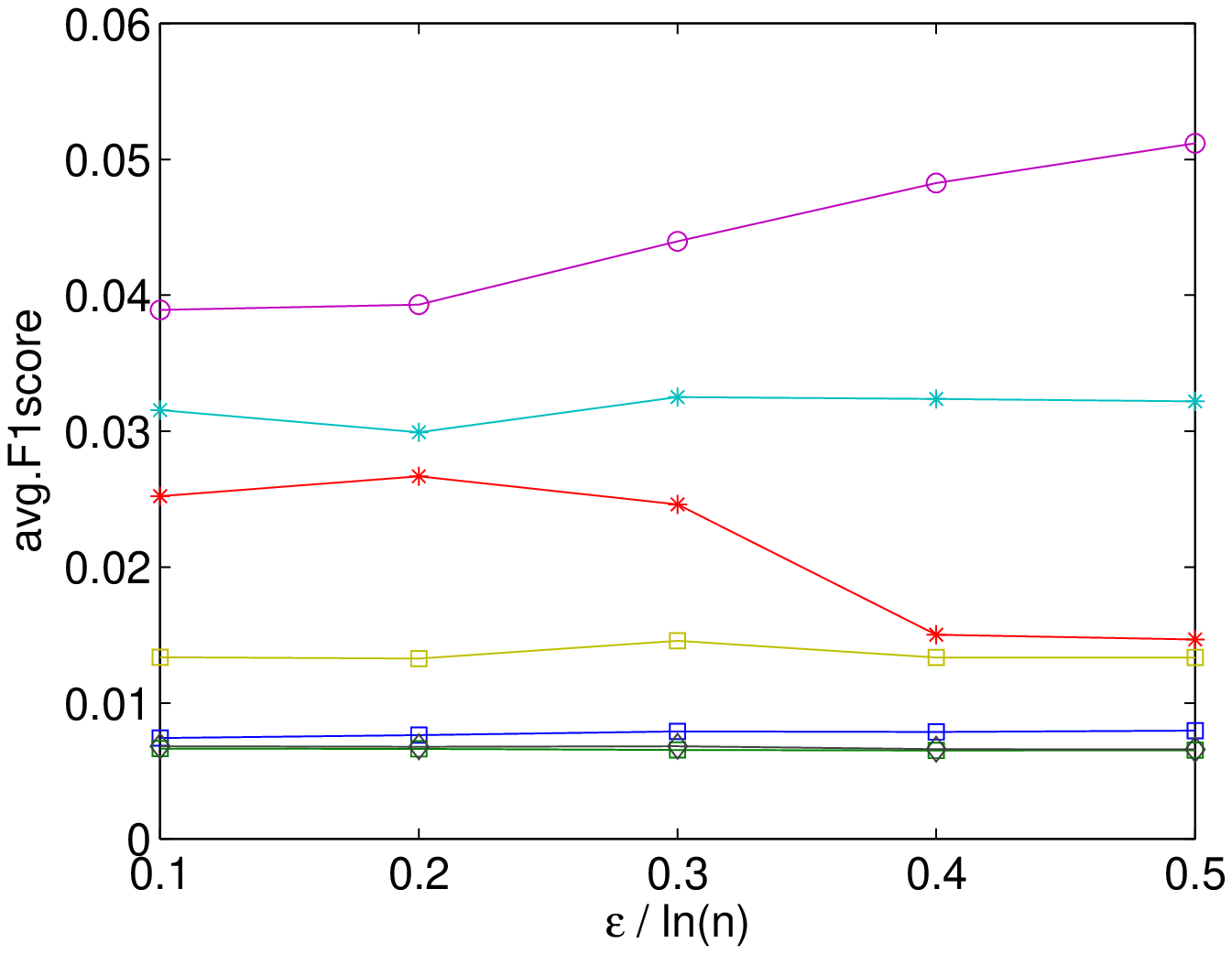}
                 \caption{}
                 \label{fig:all-amazon-f1}
         \end{subfigure}
         \hfill         
         \begin{subfigure}[t]{0.32\textwidth}
                  \centering
                  \includegraphics[height=1.6in]{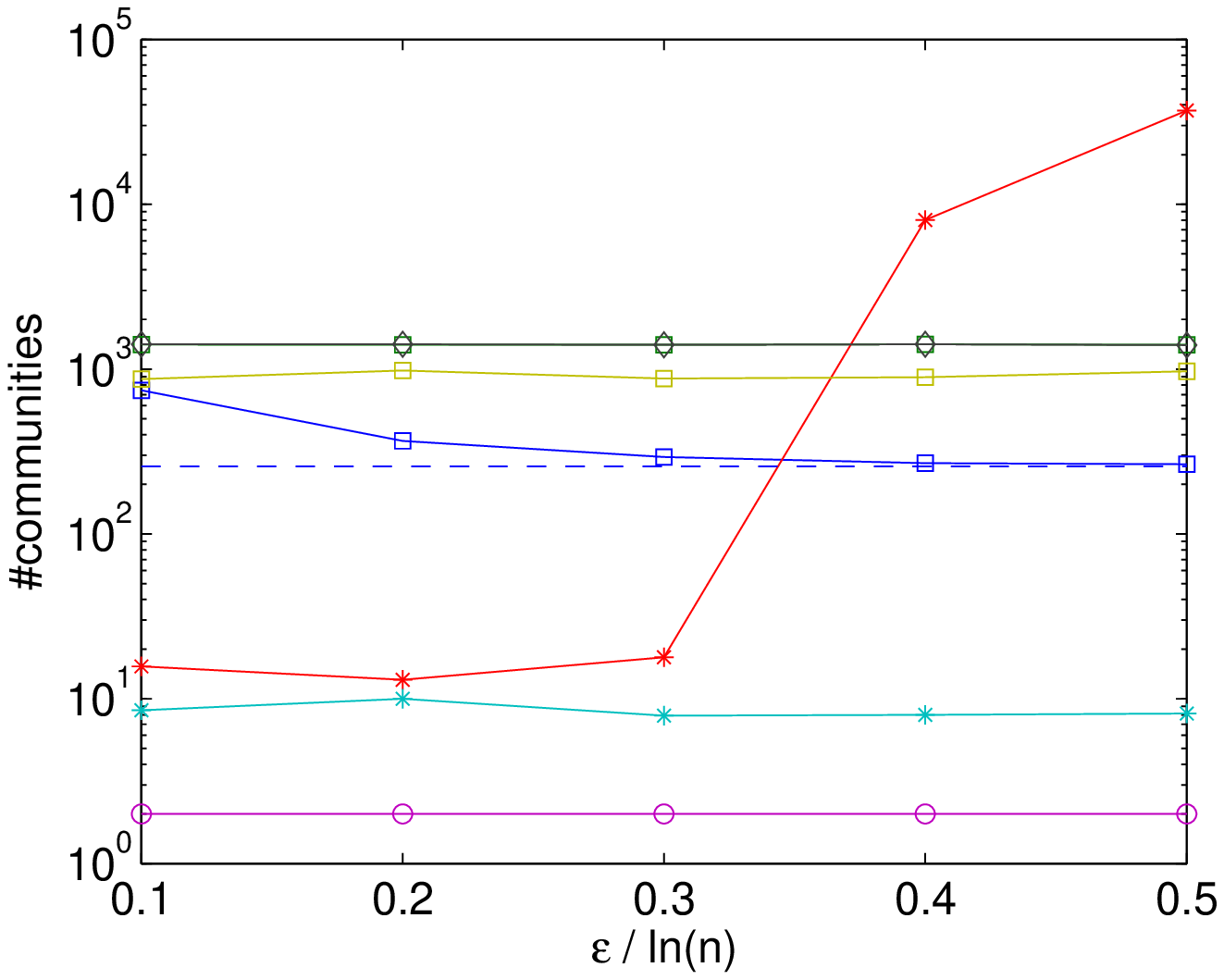}
                 \caption{}
                  \label{fig:all-amazon-com}
          \end{subfigure}        	
     \caption{Quality metrics and the number of communities (\texttt{amazon}) (0.5$\ln n$ = 6.4)}
     \label{fig:all-amazon}
 \end{figure*}

\begin{figure*}
 	\centering
         \begin{subfigure}[t]{0.32\textwidth}
                 \centering
                 \includegraphics[height=1.6in]{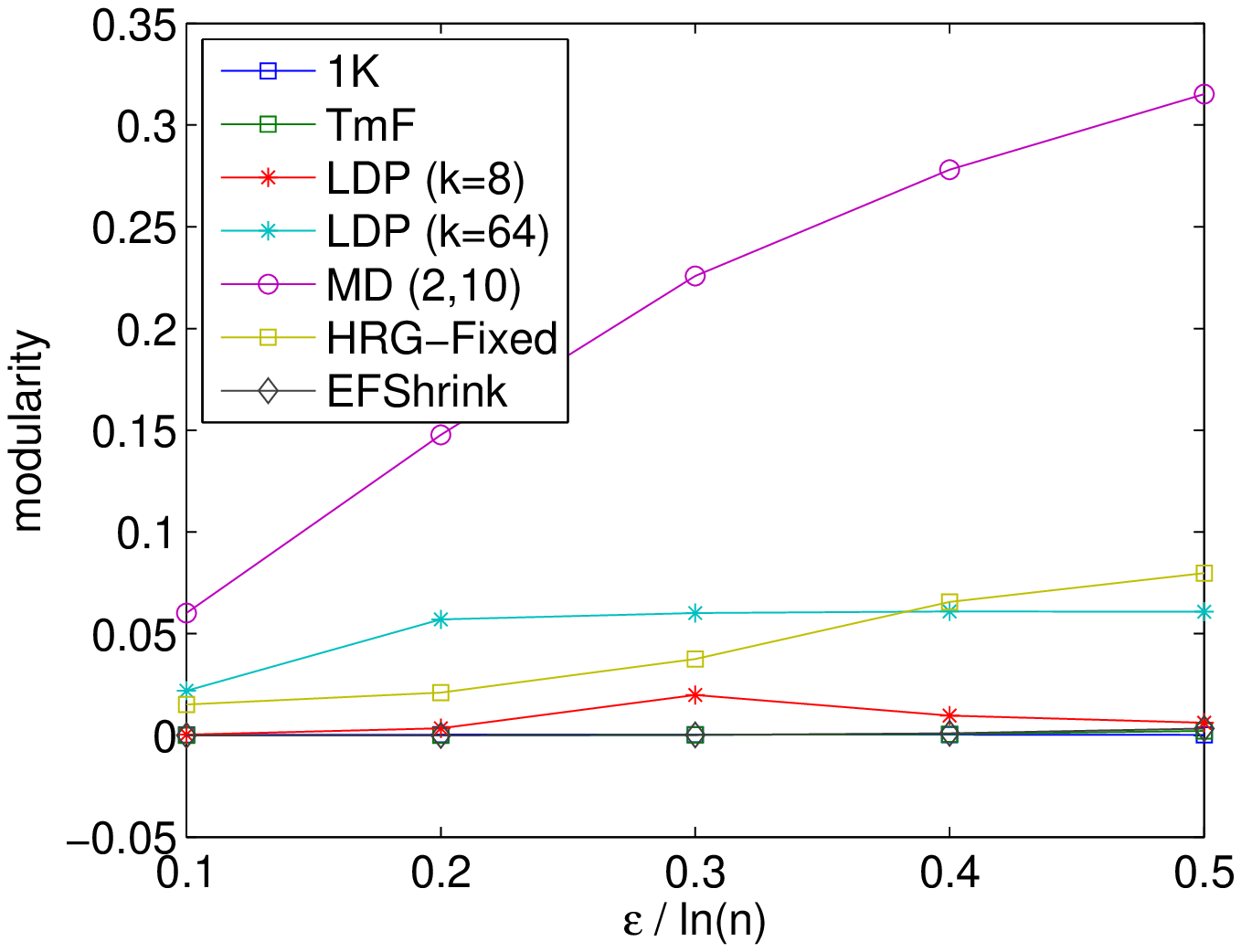}
                 \caption{}
                 \label{fig:all-dblp-mod}
         \end{subfigure}
         \hfill
         \begin{subfigure}[t]{0.32\textwidth}
                 \centering
                 \includegraphics[height=1.6in]{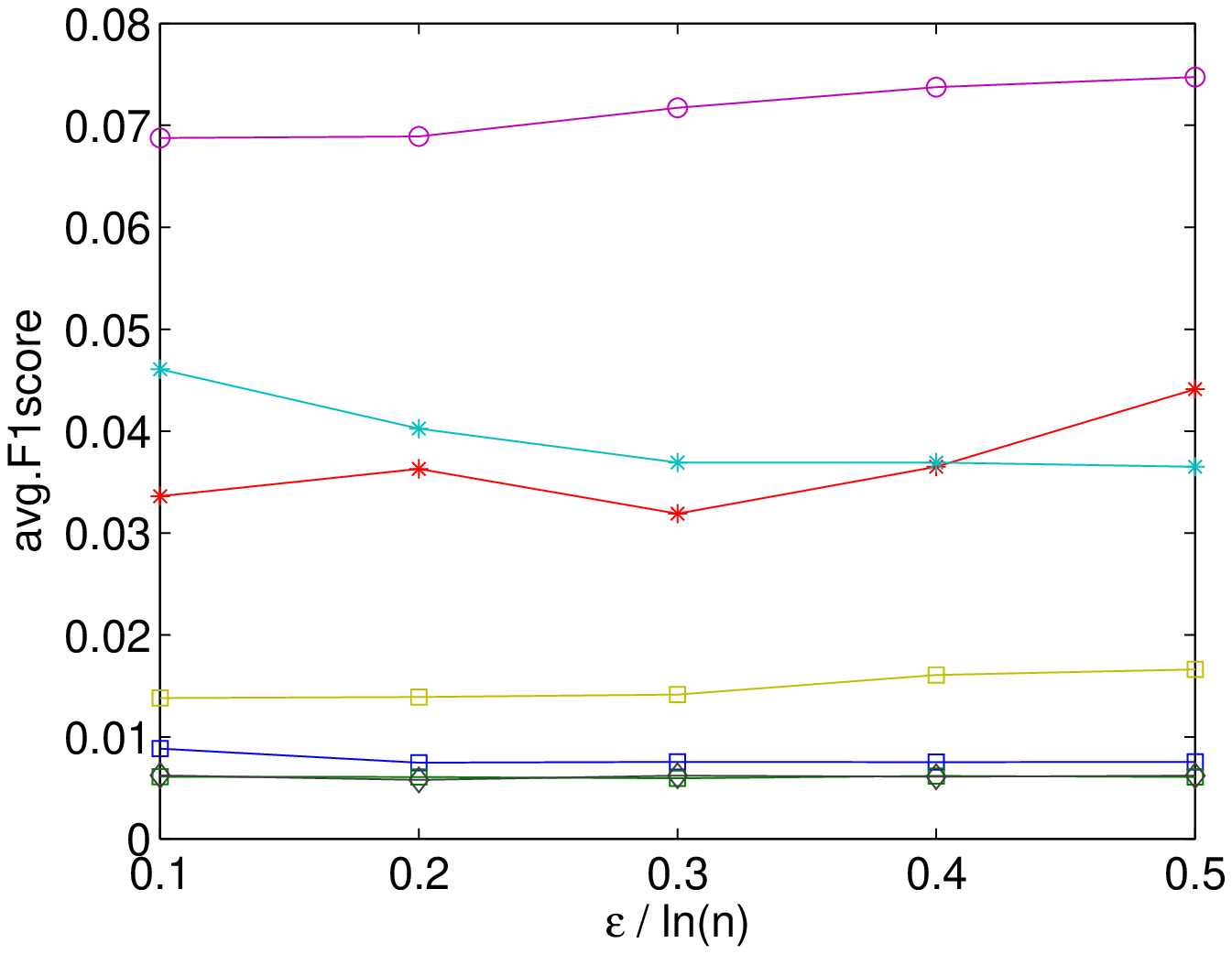}
                 \caption{}
                 \label{fig:all-dblp-f1}
         \end{subfigure}
         \hfill         
         \begin{subfigure}[t]{0.32\textwidth}
                  \centering
                  \includegraphics[height=1.6in]{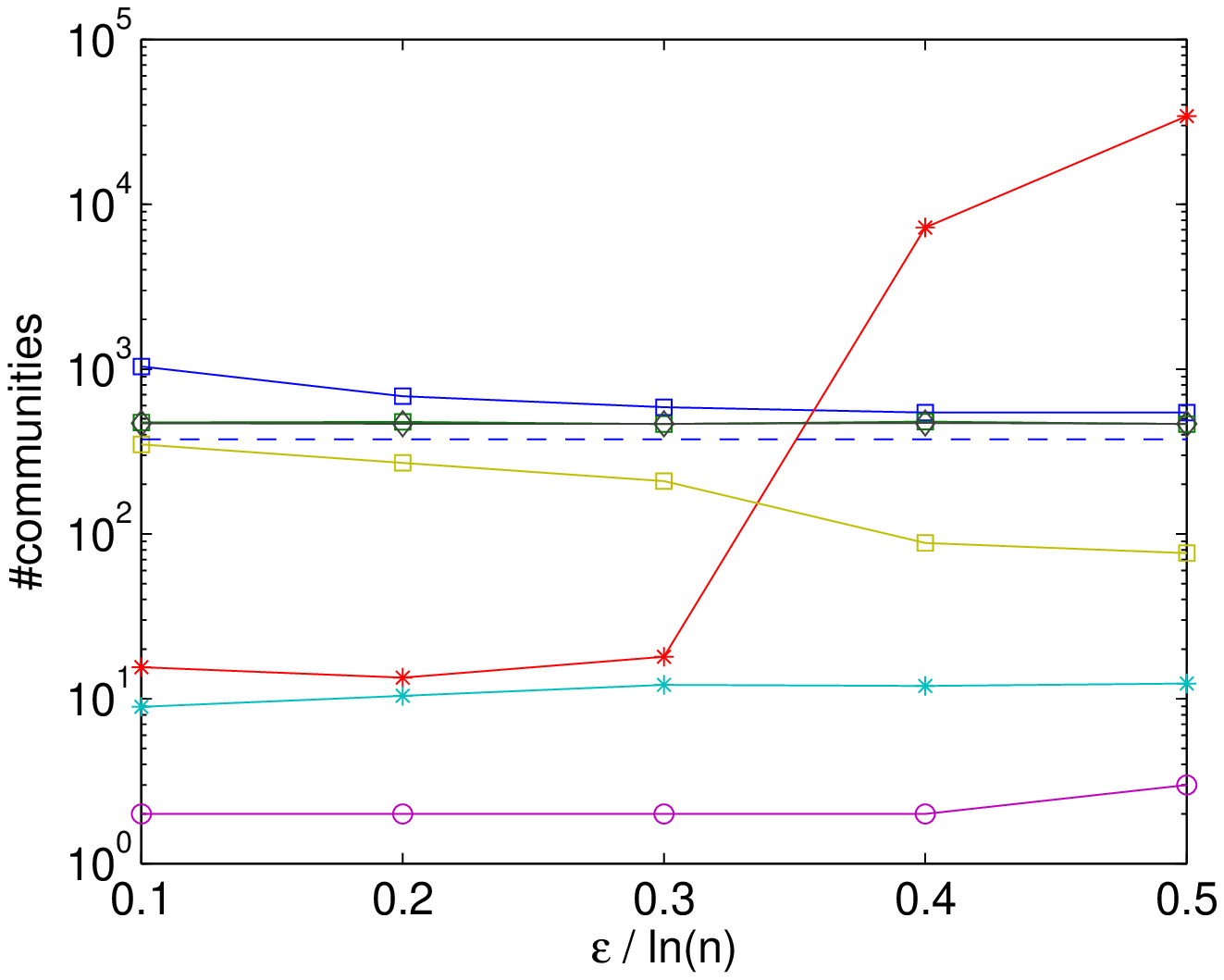}
                 \caption{}
                  \label{fig:all-dblp-com}
          \end{subfigure}        	
     \caption{Quality metrics and the number of communities (\texttt{dblp}) (0.5$\ln n$ = 6.3)}
     \label{fig:all-dblp}
 \end{figure*}

\begin{figure*}
 	\centering
         \begin{subfigure}[t]{0.32\textwidth}
                 \centering
                 \includegraphics[height=1.6in]{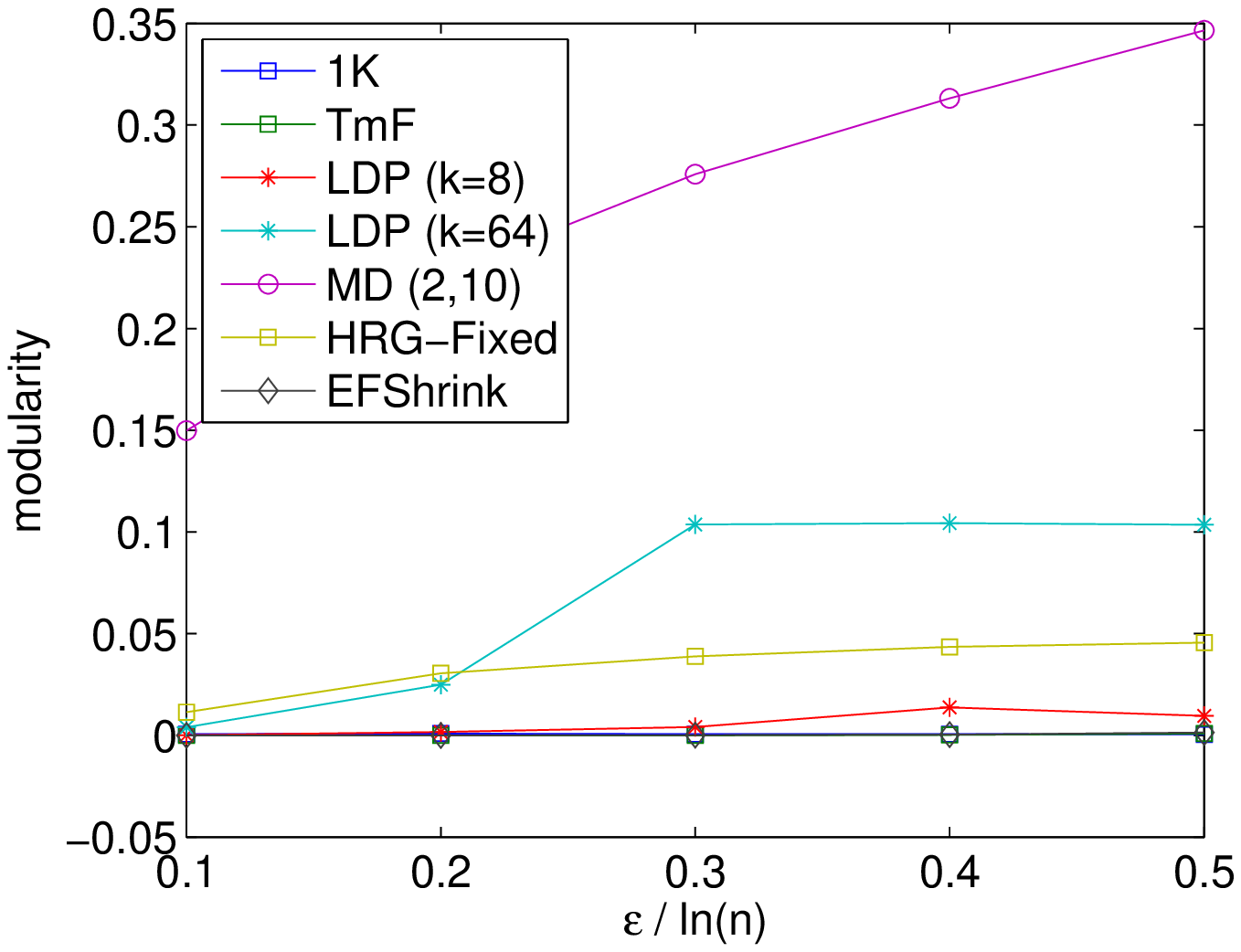}
                 \caption{}	
                 \label{fig:all-youtube-mod}
         \end{subfigure}
         \hfill
         \begin{subfigure}[t]{0.32\textwidth}
                 \centering
                 \includegraphics[height=1.6in]{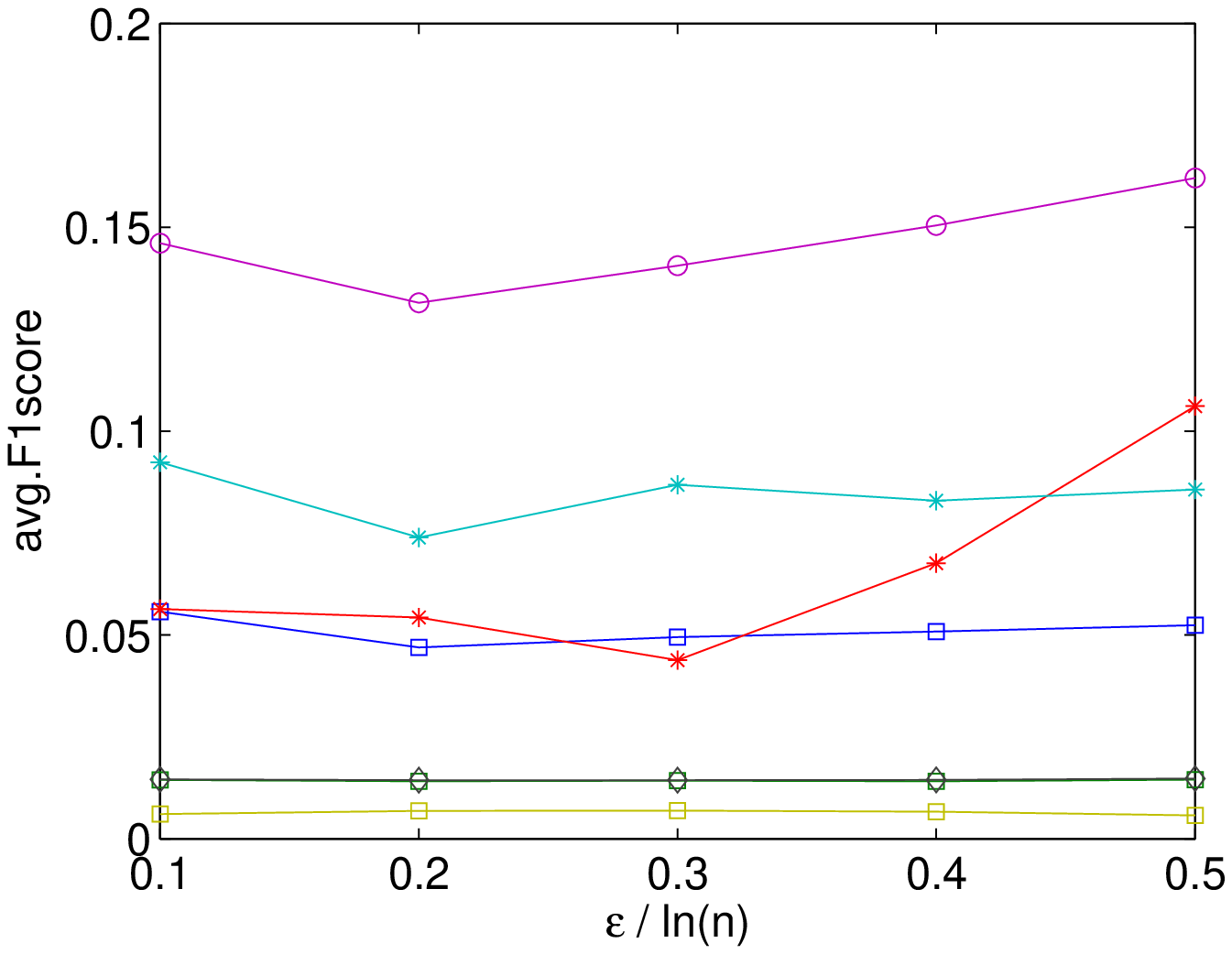}
                 \caption{}
                 \label{fig:all-youtube-f1}
         \end{subfigure}
         \hfill         
         \begin{subfigure}[t]{0.32\textwidth}
                  \centering
                  \includegraphics[height=1.6in]{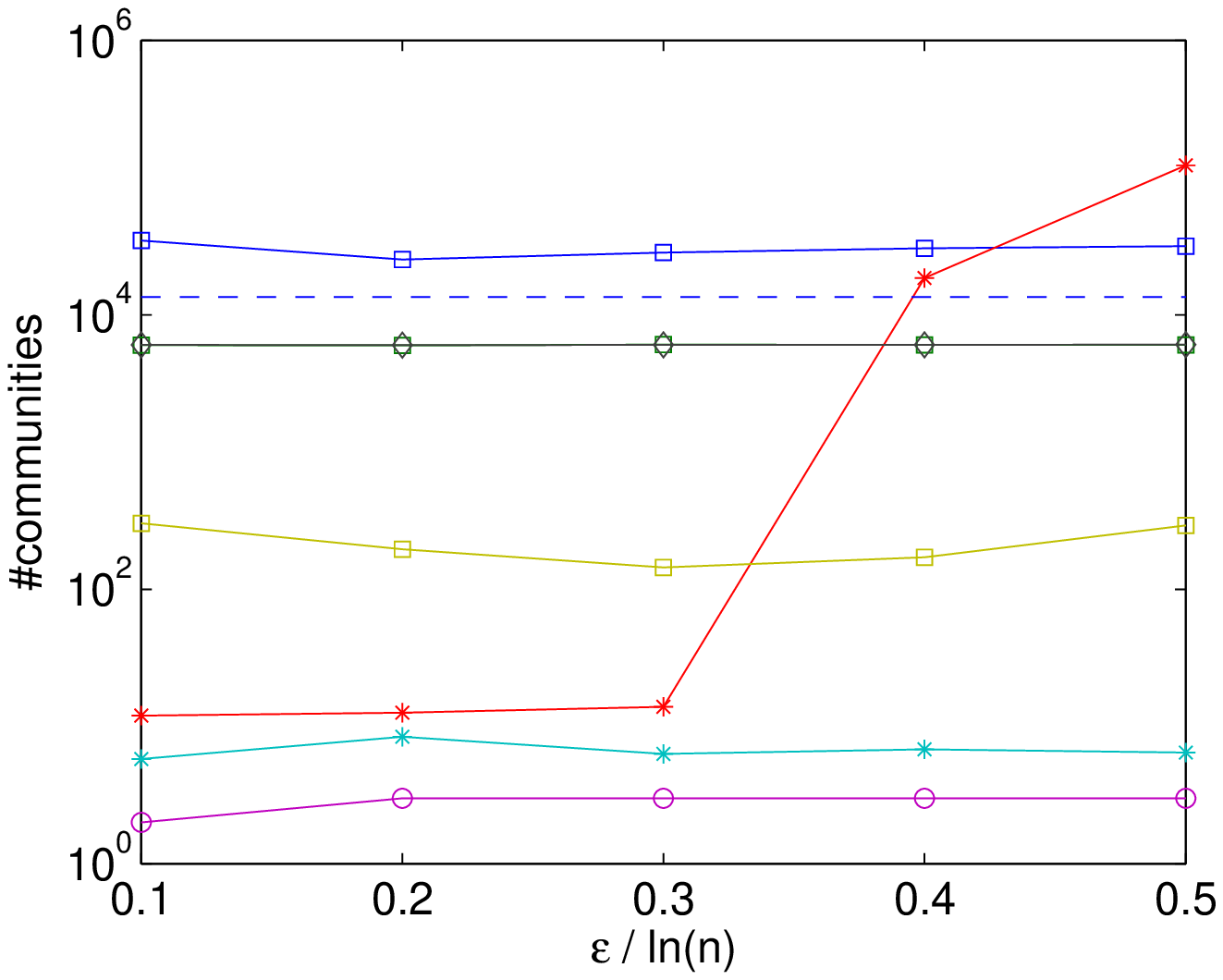}
                 \caption{}
                  \label{fig:all-youtube-com}
          \end{subfigure}        	
     \caption{Quality metrics and the number of communities (\texttt{youtube}) (0.5$\ln n$ = 7.0)}
     \label{fig:all-youtube}
 \end{figure*}



%

\bibliographystyle{abbrv}
\bibliography{community-dp-short}

\end{document}